\documentclass[11pt]{article}

\usepackage[margin=1in]{geometry}
\usepackage{authblk}
\usepackage[colorlinks=true,allcolors=blue]{hyperref}

\usepackage{graphicx}
\usepackage{amsmath}
\usepackage{amsthm}
\usepackage{amssymb}
\usepackage{mathtools}
\usepackage{natbib}
\usepackage{booktabs}
\usepackage{array}
\usepackage{multirow}
\usepackage{tabularx}
\usepackage{adjustbox}
\usepackage{subfig}
\usepackage{xcolor}
\usepackage{algorithm}
\usepackage{algpseudocode}
\usepackage{listings}
\usepackage{caption}
\graphicspath{{Fig/}}

\theoremstyle{plain}
\newtheorem{theorem}{Theorem}
\newtheorem{proposition}[theorem]{Proposition}%
\theoremstyle{definition}

\theoremstyle{remark}
\newcommand\mc{\mathcal}

\newcommand\bs{\boldsymbol}
\newcommand{\change}[1]{#1}
\begin{document}




\title{A Unified and Computationally Efficient Non-Gaussian Statistical Modeling Framework}

\author[1]{David Bolin\thanks{The authors are ordered alphabetically.}}
\author[1]{Xiaotian Jin\thanks{Corresponding author: xiaotian.jin@kaust.edu.sa}}
\author[1]{Alexandre B. Simas}
\author[2]{Jonas Wallin}

\affil[1]{CEMSE, King Abdullah University of Science and Technology, Thuwal, Saudi Arabia}
\affil[2]{Department of Statistics, Lund University, Lund, Sweden}

\date{}



\maketitle

\abstract{Datasets that exhibit non-Gaussian characteristics are common in many fields, while the current modeling framework and available software for non-Gaussian models is limited.
We introduce Linear Latent Non-Gaussian Models (LLnGMs), a unified and computationally efficient statistical modeling framework that extends a class of latent Gaussian models to allow for latent non-Gaussian processes.
The framework unifies several popular models, from simple temporal models to complex spatial-temporal and multivariate models, facilitating natural non-Gaussian extensions.
Computationally efficient Bayesian inference, with theoretical guarantees, is developed based on stochastic gradient descent estimation. The R package \texttt{ngme2}, which implements the framework, is presented and demonstrated through a wide range of applications including novel non-Gaussian spatial and spatio-temporal models.
}


\section{Introduction}\label{chp:introduction}

The development of statistical models for complex, real-world phenomena often requires balancing flexibility with computational tractability. Latent Gaussian models (LGMs) have emerged as a cornerstone of modern statistical inference, providing a powerful framework for modeling hierarchical dependencies in data. 
This framework encompasses a vast array of models including generalized linear mixed models \citep{fong2010}, survival models \citep{martino2011}, and spatial and spatio-temporal models \citep{INLA2018}. 

Their popularity stems from the combination of modeling flexibility and computational efficiency, particularly when using methods like the integrated nested Laplace approximation (INLA) \citep{rue2017}.
\change{In a Bayesian setting, a} canonical LGM is \change{defined by three layers:}
%
The first layer $\mathbf{Y}|\mathbf{W}, \bs{\theta} \sim \prod_{i\in\mc{I}} \pi(Y_i|W_i, \bs{\theta})$ relates the  data \(\mathbf{Y}\) to the latent Gaussian field \(\mathbf{W}\) via a likelihood function, which often assumes conditional independence given the latent field and parameters $\bs{\theta}$. The second layer $\mathbf{W}|\bs{\theta} \sim \mathcal{N}(\mathbf{m}(\bs{\theta}), \mathbf{Q}^{-1}(\bs{\theta}))$ defines the Gaussian field with mean \(\mathbf{m}(\bs{\theta})\) and precision matrix \(\mathbf{Q}(\bs{\theta})\), and the third $\bs{\theta} \sim \pi(\bs{\theta})$ specifies prior distributions.

Despite their versatility, LGMs have a fundamental limitation in that the Gaussian assumption imposes constraints that may be too restrictive for many real-world phenomena. 
The inherent symmetry of Gaussian processes limits their ability to capture systematic asymmetries or directional biases frequently observed in environmental, economic, and biological data.
They can further be too smooth to capture sudden jumps, regime changes, or local extremes, leading to oversmoothing and missed anomalies \citep{walder2020, dhull2021}.
These structural limitations have motivated the development of non-Gaussian extensions that preserve the computational advantages of LGMs while providing greater distributional flexibility and the ability to capture complex real-world phenomena more accurately.
In particular, the generalized hyperbolic (GH) family of distributions \citep{barndorff1978} has found widespread application across diverse statistical modeling domains. 
This is a 5-parameter distribution that is highly flexible and serves as a superclass for many important distributions, including the Gaussian distribution as a limiting case, as well as Student's $t$, Laplace, hyperbolic, normal-inverse Gaussian (NIG), and variance-gamma distributions.
Autoregressive processes with GH innovating terms are discussed in \cite{ghasami2020} and can model extreme financial market movements not captured by Gaussian models \citep{dhull2021, bibby2003}. 
In spatial statistics, \cite{bolin2014} provided a class of non-Gaussian random fields with Matérn covariance function, constructed as solutions to stochastic partial differential equations (SPDEs) driven by GH noise. 
Similar models have also been applied in the context of geostatistics \citep{wallin2015}, joint modeling of multivariate random fields \citep{multivar2020}, and to model longitudinal data \citep{asar2020}.

Despite the growing interest in non-Gaussian latent fields modeling, accessible statistical software remains limited. The widespread success of LGMs has been largely facilitated by highly efficient computational tools. For instance, the \texttt{mgcv} package \citep{wood2017} has become the standard for fitting generalized additive models (GAMs) by exploiting the equivalence between penalized smoothing splines and Gaussian random effects. Similarly, \texttt{R-INLA} \citep{rue2009, INLA2018} provides a remarkably fast framework for spatial and hierarchical modeling. 
However, these and other popular user-friendly packages like \texttt{brms} \citep{brms}, \texttt{glmmTMB} \citep{glmmTMB}, \texttt{sdmtmb} \citep{sdmtmb}, and \texttt{rstanarm} \citep{rstanarm}, are fundamentally anchored in the assumption of a latent Gaussian field. Consequently, when structural non-Gaussianity is required, users are typically forced to turn to more general and complex software like \texttt{nimble} \citep{nimble} or \texttt{Stan} \citep{stan}. These are often cumbersome for inexperienced users, and computationally expensive for high-dimensional problems \citep{cabral2024}.

\cite{asar2020} presented a framework for non-Gaussian modeling tailored to longitudinal data, accompanied by the R package \texttt{ngme} that implements a maximum likelihood method for model fitting. As this work is mainly concentrated on longitudinal data with specific temporal dependencies, it has limited applicability to data with spatio-temporal, multivariate, and other complex dependency structures.
\cite{cabral2024} proposed a similar framework in the Bayesian setting, implemented in the R package \texttt{ngvb} using variational inference \citep{jordan1999, bishop2006} to approximate posterior distributions of the model parameters. The package is an extension to \texttt{R-INLA} and combines the INLA method for fitting LGMs with the variational approach. However, the package is highly dependent on \texttt{R-INLA}, supports only a limited range of models, and lacks maintenance.

In this work, we introduce a new framework that generalizes and extends that of \cite{asar2020} to allow not only longitudinal modeling but also spatio-temporal, multivariate, and other complex dependency structures through a unified formulation, which also allows for complex non-stationary dependency structures and marginal properties.
We refer to this framework as linear latent non-Gaussian models (LLnGMs). Computational efficiency and scalability are maintained for all models in this family through mixture representations based on the GH family and  sparse matrix methods and full Bayesian inference is achieved through a stochastic gradient Langevin dynamics (SGLD) approach.
\change{Because the Gaussian distribution is a limiting case of the GH family, LLnGMs provide a natural environment for evaluating Gaussian assumptions against more flexible alternatives, ensuring that the additional complexity of a non-Gaussian model is justified by the data.}
We also present the new R package \texttt{ngme2}, which 
has full support for all introduced models.

The structure of this paper is as follows: Section~\ref{chp:framework} describes the LLnGM framework, and Section~\ref{chp:estimation} introduces the computational approaches and the inference method. Section~\ref{chp:software} provides a  brief overview of the \texttt{ngme2} R package. Section~\ref{chp:applications} presents applications of the framework in four different fields. Section~\ref{sec:comparison} compares this framework with other approaches. 
Section~\ref{chp:discussion} discusses broader implications of the framework, including its straightforward modification to frequentist estimation. 
Three appendices contain additional technical details.

\section{The LLnGM Framework}\label{chp:framework}

This section presents the LLnGM framework, which provides a unified approach to modeling complex data structures with latent non-Gaussian features.

\subsection{The model formulation}
The LLnGM framework is specified as follows:
\begin{subequations} \label{eq:framework}
\begin{align}
\text{Data:} \quad & \mathbf{Y} = \mathbf{A}\mathbf{W} + \mathbf{X}\boldsymbol{\beta} + \boldsymbol{\epsilon}^{\mathbf{Y}}, \label{eq:data}\\
\text{Process:} \quad & \mathbf{K}(\boldsymbol{\theta})\mathbf{W} = \boldsymbol{\epsilon}^{\mathbf{W}} ,\label{eq:process}\\
\text{Noise:} \quad & \boldsymbol{\epsilon}^{\star} | \mathbf{V}^{\star} \sim \mathcal{N}(\boldsymbol{\mu}^{\star}({\boldsymbol{\theta}}) \odot \left( \mathbf{V}^{\star} - \mathbf{h}^{\star} \right), \text{diag}(\boldsymbol{\sigma}^{\star}({\boldsymbol{\theta}})^2\odot \mathbf{V}^{\star})), \,\, \star\in\{\mathbf{Y},\mathbf{W}\} \label{eq:noise} \\
\text{Mixture:} \quad & V_j^{\star} \overset{\text{iid}}{\sim} \text{GIG}(p_j^{\star}(\boldsymbol{\theta}), a_j^{\star}(\boldsymbol{\theta}), b_j^{\star}(\boldsymbol{\theta})) \label{eq:distribution} \quad j = 1, \ldots, n_\star, \quad \star \in \{ \mathbf{Y},\mathbf{W} \},\\
\text{Prior:} \quad &
      \bs{\theta} \sim \pi(\bs{\theta}). \label{eq:prior}
\end{align}
\end{subequations}
The data model \eqref{eq:data} specifies how observations $\mathbf{Y}$ relate to the latent process $\mathbf{W}$ through the observation matrix $\mathbf{A}$, incorporates fixed effects $\mathbf{X} \boldsymbol{\beta}$, and accounts for possibly non-Gaussian measurement noise $\boldsymbol{\epsilon}^{\mathbf{Y}}$. 
The process model \eqref{eq:process} characterizes the dependency structure of the latent process $\mathbf{W}$, which is driven by the possibly non-Gaussian noise $\boldsymbol{\epsilon}^{\mathbf{W}}$, via the operator matrix $\mathbf{K}(\boldsymbol{\theta})$. 
As detailed in Section~\ref{subsec:unified-process}, this provides a unified framework that accommodates diverse model structures, ranging from basic random effects models to sophisticated spatio-temporal models.
The hierarchical structure is completed by \eqref{eq:prior}, which assigns prior distributions to the parameter vector $\boldsymbol{\theta}$ (including $\boldsymbol{\beta}$).
Further, \( \odot \) denotes the element-wise product, and \(\mathbf{h}\) is a fixed vector determined by the specific model such that $\mathbb{E}(\mathbf{V})=\mathbf{h}$, which makes the latent process centered. 
\change{One typically has $\mathbf{h} = \mathbf{1}$ for discrete models and $\mathbf{h}$ is determined by the discretization of continuous models (examples in Section~\ref{subsec:unified-process}).}

The noise model \eqref{eq:noise} and \eqref{eq:distribution} define the distribution of the noise $\boldsymbol{\epsilon}^{\mathbf{Y}}$ and $\boldsymbol{\epsilon}^{\mathbf{W}}$ through a normal mean-variance mixture, conditioned on the vectors $\mathbf{V}^{\mathbf{Y}}$ and $\mathbf{V}^{\mathbf{W}}$ which have independent elements following generalized inverse Gaussian (GIG) distributions. 
The \(\text{GIG}(p,a,b)\) distribution has density
\begin{equation*}
  f(x;p,a,b) = \frac{(a/b)^{p/2}}{2K_p(\sqrt{ab})} x^{p-1} \exp\left\{-\frac{1}{2}(ax + b/x)\right\}, \quad x > 0,
  \label{eq:gig}
\end{equation*}
where \(K_p\) is a modified Bessel function of the second kind of order \(p \in \mathbb{R}\), and the parameters \(a\) and \(b\) satisfy certain constraints depending on \(p\) \citep{jorgensen1982}. When \(V\sim\text{GIG}(p,a,b)\) is used in a mean-variance mixture $B = \delta + \mu V + \sigma \sqrt{V} Z$, where $\delta$, $\mu$, and $\sigma>0$ are real numbers, and \(Z \sim N(0,1)\) is a standard Gaussian random variable independent of $V$, $B$ has a GH distribution.
As previously mentioned, the GH family contains several important special cases, including Student's t, NIG, Generalized asymmetric Laplace (GAL), and Cauchy distributions. 
Table~\ref{tab:GH} summarizes these special cases and corresponding mixing distributions.



%
In \eqref{eq:noise}, the noise can be represented as
$\boldsymbol{\epsilon} = -\boldsymbol{\mu} \odot \mathbf{h} + \boldsymbol{\mu} \odot \mathbf{V} + \boldsymbol{\sigma} \odot \sqrt{\mathbf{V}} \odot \mathbf{Z}$, where $\mathbf{Z} \sim \mathcal{N}(\mathbf{0}, \mathbf{I}_n)$ and  $\mathbf{Z} \perp \mathbf{V}$.
By setting the $V_j$'s to be independent and GIG-distributed, it follows that the $\epsilon_j$'s are independent and follow GH distributions.

Note that the mean of the GIG distribution is set equal to $\mathbf{h}$, i.e., $\mathbb{E}(\mathbf{V})=\mathbf{h}$, so that the latent process is centered. Moreover, if $\mathbf{V}=\mathbf{h}$, the latent model reduces to a Gaussian model.
This mixture representation is powerful because the conditional distribution \(V_j | \epsilon_j\) remains in the GIG family, enabling efficient MCMC sampling, which is a crucial property for the efficient model inference in Section~\ref{chp:estimation}. 


\begin{table}[t]
\centering
\caption{\label{tab:GH}Special Cases of the GH Distribution and Their Corresponding Mixing Distributions}
\begin{tabular}[t]{lll}
\hline
Special case of GH distribution & Distribution of $V$ & GIG form of $V$\\
\hline
Student's t & Inverse-Gamma($\nu/2, \nu/2$) & GIG($-\nu/2, 0, \nu$)\\
Normal-inverse Gaussian  & Inverse-Gaussian($\alpha, \beta$) & GIG($-1/2, \alpha, \beta$)\\
Generalized asymmetric Laplace  & Gamma($\alpha, \beta$) & GIG($\alpha, 2\beta, 0$)\\
Cauchy & Inverse-Gamma($1/2, 1/2$) & GIG($-1/2, 0, 1$)\\
\hline
\end{tabular}
\end{table}

LLnGMs notably accommodate non-Gaussian characteristics in both the measurement noise and the latent field through independent mixing variable mechanisms. This allows for modeling of scenarios where both the underlying process and the measurement mechanism have heavy-tailed, skewed, or heteroscedastic properties.
Another key feature is the ability to naturally accommodate non-stationary behavior. 
The parameters $\boldsymbol{\mu}$, $\boldsymbol{\sigma}$, and the GIG parameters can be specified as functions of covariates, enabling location-dependent distributional characteristics. 
This non-stationarity enables the non-Gaussian characteristics to adapt to local conditions, capturing heterogeneous behavior across the domain.
One particularly interesting case enabled by the framework is models with non-stationary skewness, where $\boldsymbol{\mu}$ is a function of spatial or temporal covariates. This opens up new possibilities for modeling in several fields. 
For example, in climate modeling, precipitation skewness may vary with geographic and seasonal factors; in financial modeling, return asymmetry can depend on market conditions; and in environmental monitoring, pollutant concentration asymmetry may change with meteorological conditions.
The application of non-stationary skewness models in these domains appears to be largely unexplored, suggesting that this framework may offer novel contributions in this area. We explore this in one of the applications in Section~\ref{chp:applications}.

\subsection{Flexible but unified process models} \label{subsec:unified-process}

The versatility of the LLnGM framework becomes apparent when considering how numerous well-established statistical models can be expressed through different specifications of the operator matrix $\mathbf{K}(\boldsymbol{\theta})$. 
The framework encompasses a wide range of useful models, including but not limited to temporal models such as autoregressive models (AR(p)), random walks (RW), and the Ornstein-Uhlenbeck (OU) process. It also includes spatial models with Matérn covariance facilitated by the SPDE approach \citep{lindgren2011, wallin2015}, among others.
The exact parameterization of the $\mathbf{K}$ matrix for these models is provided in Appendix~\ref{app:operators}.
See also \citet[][Section 4]{cabral2023}, where the authors adapt the similar parameterization of the latent models, and discuss the construction of the operator matrix $\mathbf{K}(\boldsymbol{\theta})$ for different models. 

As an example, we present spatial models with Matérn covariances using the SPDE approach, as they are used in later sections.
The Matérn covariance family is 
\begin{equation}
  \text{Cov}(\mathbf{s}_1, \mathbf{s}_2) = \frac{2^{1-\nu} \sigma^2}{\Gamma(\nu)} \left(\kappa ||\mathbf{s}_1 - \mathbf{s}_2||\right)^{\nu} K_{\nu}(\kappa ||\mathbf{s}_1 - \mathbf{s}_2||),
\label{eq:matern-covariance}
\end{equation}
where $K_{\nu}$ is the modified Bessel function of the second kind, $\nu>0$ is the smoothness parameter, $\sigma^2>0$ is the marginal variance, and $\kappa>0$ controls the spatial scale (range) of dependence.
In the SPDE approach, a Gaussian Whittle--Mat\'ern field $u(\cdot)$ on a spatial domain $\mathcal{D}$ is represented as a solution to SPDE
\begin{equation}\label{eq:spde-general}
  (\kappa^2 - \Delta)^{\alpha/2} \, u(\cdot) 
  \;=\;
  \mathcal{W}(\cdot), \quad \text{on } \mathcal{D},
\end{equation}
where $\mathcal{W}(\cdot)$ is a Gaussian white noise, $\Delta$ is the Laplacian, and $\alpha = \nu + d/2$ \citep{whittle1963}.
By discretizing \eqref{eq:spde-general} on a triangulation of the domain, one obtains a model of the form \eqref{eq:process} where $\mathbf{K}$ is a sparse matrix if $\alpha$ is an even integer \citep{lindgren2011} and $\mathbf{h}$ is determined by  the triangulation, see Appendix~\ref{app:operators} for details. 
Replacing $\mathcal{W}(\cdot)$ with non-Gaussian white noise \(\dot{M}(\cdot)\) results in non-Gaussian fields with the same covariance structure and computational benefits \citep{bolin2014}. 
Both the Gaussian and non-Gaussian versions are compatible with our framework.
\change{In this example, $\mathbf{h}$ is the vector indexing the triangulation, and see the Appendix~\ref{app:operators} for the precise definition.}


By viewing the model construction through the operator matrix, one can also construct more complex models through combinations of two fundamental building blocks applied to simpler operator matrices: the \textit{multivariate} and  \textit{tensor product} operations.
The multivariate operation facilitates spatial and temporal multivariate models with cross-correlation structure.
\change{Given two LLnGMs of the form  \eqref{eq:framework} with subscript 1 and 2 to indicate the component of the associated model, we can combine them to obtain new models LLnGM of the form  \eqref{eq:framework}.}
We can introduce a bivariate process with dependencies through a dependence matrix \(\mathbf{D}\), as 
\begin{equation}
\mathbf{K} \mathbf{W} = 
(\mathbf{D}\otimes \mathbf{I}_d) \mathbf{K}_{\text{block}} \begin{pmatrix} \mathbf{W}_1 \\ \mathbf{W}_2 \end{pmatrix} = \begin{pmatrix} \boldsymbol{\epsilon}^{\mathbf{W}}_1 \\ \boldsymbol{\epsilon}^{\mathbf{W}}_2 \end{pmatrix},
\label{eq:multivariate-construction}
\end{equation}
where 
$\mathbf{K}_{\mathrm{block}} = \mathrm{blockdiag}(\mathbf{K}_1, \mathbf{K}_2)$
and $\boldsymbol{\epsilon}^{\mathbf{W}}_1$ and $\boldsymbol{\epsilon}^{\mathbf{W}}_2$ are two independent noises.
For example, \cite{multivar2020} introduce the multivariate Type G Matérn SPDE model, which parameterizes $\mathbf{D}$ through a rotation parameter $\zeta \in [0, 2\pi]$ and ``correlation'' parameter $\rho \in \mathbb{R}$ as:
\begin{equation}
\mathbf{D}(\zeta, \rho) = \begin{bmatrix}
\cos(\zeta) + \rho \sin(\zeta) & - \sin(\zeta) \sqrt{1+\rho^2} \\
\sin(\zeta) - \rho \cos(\zeta) & \cos(\zeta) \sqrt{1+\rho^2} 
\end{bmatrix}.
\label{eq:dependence-matrix}
\end{equation}
In the non-Gaussian case, we let $\mathbf{V} = \big(\mathbf{V}_1^\top, \mathbf{V}_2^\top\big)^\top$ be the vector of the mixing variables. 
Introducing $\mathbf{h} = \big(\mathbf{h}_1^\top, \mathbf{h}_2^\top\big)^\top$ and $\boldsymbol{\mu} = \big(\boldsymbol{\mu}_1^\top, \boldsymbol{\mu}_2^\top\big)^\top$, we then define the non-Gaussian model  
$\mathbf{K} \mathbf{W} \mid \mathbf{V} 
\sim \mathcal{N}\Big( \boldsymbol{\mu} \odot (\mathbf{V} - \mathbf{h}), \; \mathrm{diag}(\mathbf{V}) \Big)$ which thus fits in the LLnGM framework \eqref{eq:framework}.

The construction extends naturally to higher-dimensional cases. Further,
by setting $\mathbf{D} = \mathbf{I}_2$ and defining an observation matrix \(\mathbf{A} = [\mathbf{A}_1, \mathbf{A}_2]\) based on the observation matrices $\mathbf{A}_1$ and $\mathbf{A}_2$ for $\mathbf{W}_1$ and $\mathbf{W}_2$, we also recover the additive model of two independent processes:
\(
\mathbf{A}\mathbf{W} = \mathbf{A}_1\mathbf{W}_1 + \mathbf{A}_2\mathbf{W}_2,
\)
which allows for several independent non-Gaussian processes to be used in the latent model.

Another way to construct new processes is by using the tensor product operation. This allows for separable spatio-temporal models to be cast in this framework. 
For example, the spatio-temporal model in \cite{camel2013} is defined as 
\begin{subequations}\label{eq:spatio-temporal-model}
\begin{align}
\mathbf{y}_t &= \mathbf{z}_t \boldsymbol{\beta} + \boldsymbol{\xi}_t + \boldsymbol{\epsilon}_t, \\
\boldsymbol{\xi}_t &= \rho \, \boldsymbol{\xi}_{t-1} + \boldsymbol{\omega}_t, 
\end{align}
\end{subequations}
where
$\mathbf{y}_t = \big(y(\mathbf{s}_1, t), \ldots, y(\mathbf{s}_d, t)\big)^{\top}$,  $\mathbf{s}_i \in \mathbb{R}^2$, $t \in \mathbb{R}$,
\(\mathbf{z}_t \) is the covariate information at time \(t\),
\(\boldsymbol{\xi}_t = \big(\xi(\mathbf{s}_1, t), \ldots, \xi(\mathbf{s}_d, t)\big)^{\top}\) is the latent process,
$\boldsymbol{\epsilon}_t \sim \mathcal{N}\big( \mathbf{0}, \sigma_\epsilon^2 \mathbf{I}_d \big)$ is measurement noise,
and $\boldsymbol{\omega}_t$ are independent copies of a discretized SPDE field, satisfying $\boldsymbol{K}_{\text{Mat\'ern--SPDE}}(\boldsymbol{\theta})\boldsymbol{\omega}_t = \boldsymbol{\epsilon}_t^W$ with $\boldsymbol{\epsilon}_t^W \sim \mathcal{N}(\mathbf{0},\text{diag}(\boldsymbol{h}))$. 
%
This model fits into the LLnGM framework \eqref{eq:framework} by setting
$\mathbf{Y}=(\mathbf{y}_1^\top,\ldots,\mathbf{y}_T^\top)^\top$,
$\mathbf{W}=(\boldsymbol{\xi}_1^\top,\ldots,\boldsymbol{\xi}_T^\top)^\top$,
$\boldsymbol{\epsilon}^{\mathbf{Y}}=(\boldsymbol{\epsilon}_1^\top,\ldots,\boldsymbol{\epsilon}_T^\top)^\top$, $\mathbf{X} = (\boldsymbol{z}_1^\top,\ldots,\boldsymbol{z}_T^\top)^\top$, 
$\boldsymbol{\epsilon}^{\mathbf{W}}=(\boldsymbol{\omega}_1^\top,\ldots,\boldsymbol{\omega}_T^\top)^\top$ and $
\mathbf{A} = \mathbf{I}_T \otimes \mathbf{I}_d .
$
Moreover, the spatio-temporal precision operator can be written as the tensor product of the AR(1) and Mat\'ern--SPDE operators,
$
\mathbf{K}(\boldsymbol{\theta})
= \mathbf{K}_{\text{AR1}} \otimes \mathbf{K}_{\text{Mat\'ern--SPDE}} .$
This allows for non-Gaussian spatio-temporal field modeling by replacing $\boldsymbol{\epsilon}^Y$ and $\boldsymbol{\epsilon}^W$ with GH noise.
\change{Details are provided in Appendix~\ref{app:operators}.}

The tensor product formulation also makes it straightforward to construct new models by replacing the base operators with alternatives. 
As a special case, a replicate (random-effects) structure can be represented by
\(
\mathbf{K} = \mathbf{I}_r \otimes \mathbf{K}_0,
\)
where $r$ denotes the number of replicates and $\mathbf{K}_0$ is the base operator shared across replicates.


These multivariate and tensor operations can also be combined to construct more complex models.
The model in \eqref{eq:spatio-temporal-model} can be extended to bivariate cross-correlated spatial fields
by applying the multivariate operation to the spatial operator while retaining the temporal AR(1) structure,
\begin{equation}
\mathbf{K}(\boldsymbol{\theta})
=
\mathbf{K}_{\text{AR1}}(\rho)
\otimes
\Big(
(\mathbf{D} \otimes \mathbf{I}_d)\;\mathrm{blockdiag}(\mathbf{K}_{\text{1}}(\boldsymbol{\theta}_1),\mathbf{K}_{\text{2}}(\boldsymbol{\theta}_2))
\Big),
\label{eq:tensor-multivariate-construction}
\end{equation}
where $\mathbf{D} \in\mathbb{R}^{2\times 2}$ is a dependence (mixing) matrix that couples the two spatial components.
To describe the resulting model, extend the observation vector to length $2d$ by stacking the two processes.
With $\mathbf{y}_t^{(u)},\mathbf{y}_t^{(v)}\in\mathbb{R}^d$ (and similarly for $\mathbf{z}_t$, $\boldsymbol{\xi}_t$, $\boldsymbol{\omega}_t$, and $\boldsymbol{\epsilon}_t$), define
\begin{equation*}
\mathbf y_t = \begin{pmatrix} \mathbf y_t^{(u)} \\ \mathbf y_t^{(v)} \end{pmatrix},\quad
\mathbf z_t = \begin{pmatrix} \mathbf z_t^{(u)} \\ \mathbf z_t^{(v)} \end{pmatrix},\quad
\boldsymbol{\xi}_t = \begin{pmatrix} \boldsymbol{\xi}_t^{(u)} \\ \boldsymbol{\xi}_t^{(v)} \end{pmatrix},\quad
\boldsymbol{\omega}_t = \begin{pmatrix} \boldsymbol{\omega}_t^{(u)} \\ \boldsymbol{\omega}_t^{(v)} \end{pmatrix},\quad
\boldsymbol{\epsilon}_t = \begin{pmatrix} \boldsymbol{\epsilon}_t^{(u)} \\ \boldsymbol{\epsilon}_t^{(v)} \end{pmatrix}.
\end{equation*}

We can then define the Gaussian spatio-temporal model through \eqref{eq:spatio-temporal-model}, where now $\boldsymbol{\epsilon}_t \sim \mathcal{N}\Big(
 \mathbf{0},
 \mathrm{blockdiag}(\sigma^2_{\epsilon^{(u)}} \mathbf{I}_d,\; \sigma^2_{\epsilon^{(v)}} \mathbf{I}_d)
\Big)$ and $\boldsymbol{\omega}_t$ are independent multivariate fields defined through \eqref{eq:multivariate-construction}.
To extend the model to non-Gaussian settings, we simply use the non-Gaussian extension of the multivariate model in \eqref{eq:multivariate-construction} for each $\mathbf{\omega}_t$. Having all mixing variables independent across locations $j$ and time $t$ and stacking the variables over time, this model can again be written in the LLnGM framework \eqref{eq:framework}.

We can also construct another bivariate spatio-temporal model by switching the order of the tensor product and the multivariate operation.
That is,
\begin{equation}
\mathbf{K} = \text{blockdiag}(\mathbf{K}^{(u)}_{\text{AR1}} \otimes \mathbf{K}^{(u)}_{\text{Matérn}}, \mathbf{K}^{(v)}_{\text{AR1}} \otimes \mathbf{K}^{(v)}_{\text{Matérn}}) \cdot \mathbf{D}.
\label{eq:multivariate-tensor-construction}
\end{equation}
In the Gaussian case, this results in the model \cite{lenzi2020} used to model wind vectors. The difference between this and the former construction is that the former is an AR(1) model in time driven by correlated bivariate fields, whereas this is a bivariate coregionalization model based on two independent AR(1) models. 

While separable space--time constructions (e.g., $\mathbf{K}_t \otimes \mathbf{K}_s$) act as effective baselines, they generally fail to represent genuine space--time interactions, such as transport or directionality, where temporal evolution depends on spatial gradients \citep{lindgren2023diffusion, advection_diffusion}. The LLnGM framework can also incorporate non-separable models derived from evolution SPDEs, such as the advection--diffusion equation. 
Discretizing such models using implicit Euler time-stepping and FEM in space yields a sparse block-lower-bidiagonal operator matrix 
$\mathbf{K}_{\text{adv-diff}}$,
where the blocks encodes the temporal dependence and the spatial operator (including diffusion and advection). 
The details of the construction are discussed in Appendix~\ref{app:operators}.
Further, just as we extended separable models to the bivariate case in \eqref{eq:multivariate-tensor-construction}, we can construct bivariate \emph{non-separable} models. By replacing the separable blocks in \eqref{eq:multivariate-tensor-construction} with non-separable advection-diffusion operators, we obtain
$\mathbf{K}_{\text{bi-nonsep}} = \text{blockdiag}(\mathbf{K}_{\text{adv-diff}}^{(u)}, \mathbf{K}_{\text{adv-diff}}^{(v)}) \cdot \mathbf{D}$.
This construction retains the sparsity required for efficient inference while allowing for modeling of complex, coupled, and physically-motivated non-separable processes.

Thus, the LLnGM framework unifies models ranging from simple temporal time series models to complex multivariate non-separable spatio-temporal processes.

\section{Inference and prediction} \label{chp:estimation}

\subsection{Parameter Estimation}
\label{sec:map}

\change{Maximum a posteriori (MAP)} estimation of LLnGMs involves maximizing the \change{log-posterior} function \change{$\ell(\boldsymbol{\theta}; \mathbf{Y}) + \log \pi(\boldsymbol{\theta})$}, \change{or equivalently minimizing the negative log-posterior (energy) function $E(\boldsymbol{\theta})$}.
Due to the non-Gaussian hierarchical model structure, the direct computation of the likelihood component is intractable.
However, we can approximate the gradient and Hessian of the objective function, which facilitates the use of both gradient-based first-order optimization methods and Newton-type second-order optimization methods.
A variety of first-order optimization methods can be utilized for maximizing the posterior, such as stochastic gradient descent (SGD) \citep{robbins1951}, which was used by \cite{asar2020} for a subclass of LLnGMs, ADAM \citep{adam}, and ADAMW \citep{adamw}, among others. For second-order methods, Newton's method and natural gradient descent, which employ the Fisher information matrix as a preconditioner \citep{amari1998}, can be applied.

Since the marginal log-likelihood \(\ell(\boldsymbol{\theta}; \mathbf{Y})\) is intractable for LLnGM models, we employ Fisher's identity \citep{douc2013nonlinear} to express the gradient in terms of the augmented log-likelihood, 
    $\nabla \ell(\boldsymbol{\theta}; \mathbf{Y}) = \mathbb{E}_{\mathbf{V}, \mathbf{W}} \left[\nabla \ell(\boldsymbol{\theta}; \mathbf{Y}, \mathbf{V}, \mathbf{W}) \mid \mathbf{Y} \right]$,
where \(\ell(\boldsymbol{\theta}; \mathbf{Y}, \mathbf{V}, \mathbf{W})\) is the augmented log-likelihood that can be computed explicitly since it involves only the conditional distributions \(\pi(\mathbf{W}|\mathbf{V})\), \(\pi(\mathbf{Y}|\mathbf{V}, \mathbf{W})\), and \(\pi(\mathbf{V})\).
The expectation in the gradient expression can be approximated using Monte Carlo integration. Using samples \(\{(\mathbf{V}^{(j)}, \mathbf{W}^{(j)})\}_{j=1}^k\) from the posterior distribution, the \change{stochastic} gradient estimator \change{for the objective function $E(\boldsymbol{\theta})$} is
\begin{equation}
    \change{\nabla E(\boldsymbol{\theta}) \approx \mathbf{g}(\boldsymbol{\theta}) = - \frac{1}{k} \sum_{j=1}^k \nabla \ell(\boldsymbol{\theta}; \mathbf{Y}, \mathbf{V}^{(j)}, \mathbf{W}^{(j)}) - \nabla \log \pi(\boldsymbol{\theta}),}
    \label{eq:mc-gradient}
\end{equation}
which is unbiased by construction, and has a variance which decreases as \(O(1/k)\). The exact form of the gradient estimator is provided in Appendix~\ref{app:inference-details}.

Given the gradient estimator, we can then directly use any first-order method. For example, the SGD method updates the parameter vector iteratively according to
$\boldsymbol{\theta}^{(t+1)} = \boldsymbol{\theta}^{(t)} - \gamma_t \mathbf{g}(\boldsymbol{\theta}^{(t)})$,
where \(\gamma_t > 0\) is the step size at iteration \(t\). \change{To ensure that constraints on the parameters are satisfied during optimization, we typically perform the optimization on a transformed scale (e.g., using a logarithmic transform for positive parameters) so that the parameters that are optimized are unconstrained.} 

The LLnGM structure involves unobserved mixing variables \(\mathbf{V} = (\mathbf{V}^W, \mathbf{V}^Y)\) and latent fields \(\mathbf{W}\). 
Crucially, each conditional distribution belongs to a known parametric family, enabling efficient Gibbs sampling. Specifically, \(\mathbf{W} | \mathbf{V}, \mathbf{Y}\) follows a multivariate Gaussian distribution, and both \(V^W_i | \mathbf{W}, \mathbf{Y}\) and \(V^Y_i | \mathbf{W}, \mathbf{Y}\) follow GIG distributions. The specific forms of these conditional distributions are provided in Appendix~\ref{app:inference-details}. Importantly, the precision matrix of \(\mathbf{W} | \mathbf{V}, \mathbf{Y}\) is sparse, which facilitates efficient sampling.
Algorithm~\ref{algo:gibbs} presents the Gibbs sampling algorithm used in the gradient estimator.

\begin{center}
\begin{minipage}{\linewidth}
\begin{algorithm}[H]
\caption{Gibbs Sampling for LLnGM Models}\label{algo:gibbs}
\begin{algorithmic}[1]
\Require Data $\mathbf{Y}$, parameters $\boldsymbol{\theta}$, number of iterations $T$
\State Initialize $\mathbf{V}^{(0)} = (\mathbf{V}^{W,(0)}, \mathbf{V}^{Y,(0)})$
\For{$t = 1, 2, \ldots, T$}
    \State Sample $\mathbf{W}^{(t)} \sim p(\mathbf{W} | \mathbf{V}^{(t-1)}, \mathbf{Y})$ \Comment{Multivariate Gaussian}
    \State Sample $\mathbf{V}^{W,(t)} \sim p(\mathbf{V}^W | \mathbf{W}^{(t)}, \mathbf{Y})$ \Comment{GIG distributions}
    \State Sample $\mathbf{V}^{Y,(t)} \sim p(\mathbf{V}^Y | \mathbf{W}^{(t)}, \mathbf{Y})$ \Comment{GIG distributions}
    \State Set $\mathbf{V}^{(t)} = (\mathbf{V}^{W,(t)}, \mathbf{V}^{Y,(t)})$
\EndFor
\State \Return $\{(\mathbf{V}^{(t)}, \mathbf{W}^{(t)})\}_{t=1}^T$ \Comment{Samples from posterior $p(\mathbf{V}, \mathbf{W} | \mathbf{Y})$}
\end{algorithmic}
\end{algorithm}
\end{minipage}
\end{center}

While the Monte Carlo estimator \eqref{eq:mc-gradient} is unbiased, its variance can be substantial, leading to noisy gradient estimates and slow convergence of the optimization algorithm. We can reduce this variance using a Rao-Blackwellization technique \citep{gelfand1990, liu1994}, which exploits the analytical tractability of certain conditional expectations.
In our context, we can apply Fisher's identity a second time to obtain
    $\nabla \ell(\boldsymbol{\theta}; \mathbf{Y}, \mathbf{V}) 
    = \mathbb{E}_{\mathbf{W}} \left[\nabla \ell(\boldsymbol{\theta}; \mathbf{Y}, \mathbf{V}, \mathbf{W}) \mid \mathbf{Y}, \mathbf{V} \right]$.
Therefore, we can construct the Rao-Blackwellized estimator
    \change{$\mathbf{g}_{\text{RB}}(\boldsymbol{\theta}) = - \frac{1}{k} \sum_{j=1}^k \nabla \ell(\boldsymbol{\theta}; \mathbf{Y}, \mathbf{V}^{(j)}) - \nabla \log \pi(\boldsymbol{\theta})$},
using only the mixing variable samples \(\{\mathbf{V}^{(j)}\}_{j=1}^k\).
The theoretical foundation for this approach relies on the ergodicity of the Gibbs sampler. Under the mild conditions established in the recent theoretical work \citet{ergodicity}, the Gibbs chain \(\{\mathbf{V}^{(j)}\}\) is geometrically ergodic, which combined with appropriate integrability conditions, ensures that
$\mathbf{g}_{\text{RB}}(\boldsymbol{\theta}) \rightarrow \nabla E(\boldsymbol{\theta})$ a.s.~and in $L^1$ by Birkhoff's ergodic theorem.
This convergence result establishes that the Rao-Blackwellized estimator is both unbiased and consistent. The geometric ergodicity further provides exponential convergence rates and enables the use of small batch sizes in the optimization algorithm, making the method computationally efficient.
\change{In practice, we observe that preconditioned SGD using Rao-Blackwellization with 5 Gibbs samples per iteration provides  fast and stable optimization in most situations. An example of traceplots for the optimizer is given in Appendix~\ref{app:comparison}.}

\subsection{Convergence Diagnostics}
\label{sec:convergence-diagnostics}

\change{To assess the convergence of the optimization process, we initialize multiple independent chains from different starting locations in parallel. Adopting a checkpoint-based monitoring strategy, we evaluate stability at regular intervals (e.g., every 10 iterations) using three complementary diagnostics.}

First, 
we evaluate between-chain agreement using the
Gelman--Rubin statistic $\widehat{R}$ \citep{gelman1992}, which compares variability across chains to variability within chains over recent checkpoints. Values of $\widehat{R}$ close to one indicate that different initializations lead to similar parameter trajectories.
Second, we assess within-chain stabilization by examining recent checkpoint summaries for each chain.
Convergence is supported when parameters fluctuate only mildly around a stable level and exhibit
negligible systematic drift. Practically, this can be quantified via the relative variability of
checkpoint means and the magnitude of a fitted linear trend over a sliding window \citep{asar2020}.
Third, we complement parameter-based criteria with a gradient-based diagnostic.
Let $\widehat{\mathbf{g}}_t$ denote the stochastic gradient estimate at iteration $t$; following
\citet{ljung1992,pesme2020convergence}, we consider the accumulated inner products
${S_T=\sum_{t=1}^{T-1}\widehat{\mathbf{g}}_t^\top \widehat{\mathbf{g}}_{t+1}}$.
Persistent directional movement typically yields a large accumulation, whereas near-equilibrium
stochastic-approximation dynamics lead to small values of $S_T$.

These diagnostics provide complementary evidence of convergence by combining cross-chain
agreement, within-chain stabilization, and gradient-level behavior.

\subsection{Full Bayesian Inference}

While MAP estimation provides a computationally efficient point estimate of the parameters, it does not capture the full posterior uncertainty. To address this, we perform full Bayesian inference using stochastic gradient Langevin dynamics (SGLD) \citep{welling2011}. SGLD seamlessly integrates with our existing gradient-based optimization framework by adding a scaled Gaussian noise term to the standard SGD update rule. This allows us to transition from optimization to posterior sampling using the same estimated gradients derived in \eqref{eq:mc-gradient}.

The inference procedure consists of two phases. In the first phase, we run the optimization algorithm described in Section~\ref{sec:map} until the parameters stabilize around the MAP estimate. This phase effectively serves as the ``burn-in'' period for the MCMC chain. Once convergence is detected via the diagnostics in Section~\ref{sec:convergence-diagnostics}, we switch to the SGLD phase. The parameter update rule at iteration $t$ becomes:
\begin{equation}
    \boldsymbol{\theta}^{(t+1)} = \boldsymbol{\theta}^{(t)} - \gamma_t \mathbf{g}(\boldsymbol{\theta}^{(t)}) + \boldsymbol{\phi}^{(t)}, \quad \boldsymbol{\phi}^{(t)} \sim \mathcal{N}(\mathbf{0}, 2\gamma_t \mathbf{I}),
    \label{eq:sgld-update}
\end{equation}
where $\mathbf{g}(\boldsymbol{\theta}^{(t)})$ is the stochastic gradient of the negative log-posterior (energy function), and $\gamma_t$ is the step size. The injected noise $\boldsymbol{\phi}^{(t)}$ ensures that the trajectory of $\boldsymbol{\theta}$ converges to the posterior distribution $\pi(\boldsymbol{\theta} | \mathbf{Y}) \propto \exp(-E(\boldsymbol{\theta}))$ as $t \to \infty$ and $\gamma_t \to 0$ \citep{welling2011, ma2015}.

This approach is particularly advantageous for LLnGMs because it leverages the efficient, sparse-matrix-based gradient computation and Rao-Blackwellization developed for the MAP estimation, avoiding the need for computationally expensive Metropolis-Hastings acceptance steps.

\subsection{Model Prediction}
\label{sec:model-prediction}

Prediction is performed by propagating posterior uncertainty from the fitted LLnGM to the
quantities of interest. The default target is the latent linear predictor, which supports spatial and
temporal prediction with coherent uncertainty quantification.

Posterior uncertainty is represented through samples of the latent variables from a Gibbs
sampler targeting the augmented posterior distribution
$p_{\hat\theta}(\mathbf{Z}\mid \mathbf{Y})$, where
$\mathbf{Z}=(\mathbf{W},\mathbf{V}^Y,\mathbf{V}^W)$. Concretely, we run the Gibbs sampler described in
Algorithm~\ref{algo:gibbs} and collect draws
$\mathbf{Z}^{(j)} \sim p_{\hat\theta}(\mathbf{Z}\mid \mathbf{Y})$, $j=1,\ldots,k$.

Let $\boldsymbol{\eta} := \mathbf{A}\mathbf{W} + \mathbf{X}\boldsymbol{\beta}$ denote the linear predictor at
observed indices, and define the predictor at new indices (or locations) as
\(
\boldsymbol{\eta}_\star := \mathbf{A}_\star \mathbf{W} + \mathbf{X}_\star \boldsymbol{\beta},
\)
where $\mathbf{A}_\star$ and $\mathbf{X}_\star$ are the corresponding observation operator and covariates for the
prediction set. Given Gibbs draws
$\mathbf{Z}^{(j)}$,
predictive samples of the linear predictor are obtained directly via
\(
\boldsymbol{\eta}_\star^{(j)} = \mathbf{A}_\star \mathbf{W}^{(j)} + \mathbf{X}_\star \boldsymbol{\beta},
\quad j=1,\ldots,k.
\)
Posterior summaries (e.g., point predictions, credible intervals, exceedance probabilities) are then
computed from $\{\boldsymbol{\eta}_\star^{(j)}\}_{j=1}^k$.


\section{The \texttt{ngme2} Package: Implementation of the Framework}
\label{chp:software}

The R package \texttt{ngme2} \citep{ngme2} provides a comprehensive implementation of the LLnGM framework introduced in Section~\ref{chp:framework}. 
It translates the formulation in \eqref{eq:framework} into a set of computational tools and implements the
stochastic gradient based estimation methodology described in Section~\ref{chp:estimation}. 
The package combines an R interface for model specification with a high performance C++ backend. 
It supports parallel chain estimation using shared-memory multithreading via OpenMP \citep{dagum1998openmp}, and computationally intensive operations are accelerated using optimized linear algebra routines, including Eigen
\citep{eigenweb} for matrix computations and sparse solvers such as CHOLMOD \citep{chen2008} and
PARDISO \citep{pardiso} for efficient sparse matrix factorization.

Noise components specify the distribution of the driving noise in the latent process and the
measurement noise. In \texttt{ngme2}, common choices are provided through constructors such as
\texttt{noise\_normal(sigma)} for Gaussian noise, \texttt{noise\_nig(mu, sigma, nu)} for
NIG noise, and \texttt{noise\_gal(mu, sigma, nu)} for GAL noise.

Latent processes in the LLnGM framework are specified using the \texttt{f()} function, following a
syntax similar to that of \texttt{R-INLA}. The general form is
\texttt{f(map, model = <type>(mesh = mesh, ...), noise = noise, \ldots)}, where \texttt{<type>} denotes
a model class (e.g., \texttt{ar1}, \texttt{matern}). 
Internally, \texttt{f()} maps the model specification to the corresponding operator matrices \(\mathbf{K}(\boldsymbol{\theta})\), and manages sparse matrix construction, parameter initialization, and model specific constraints. In particular, it supports tensor product and multivariate constructions as discussed in Section~\ref{subsec:unified-process}.

Standard \textsf{R} formula syntax is used to define the linear predictor, including fixed effects and
multiple latent components through repeated \texttt{f()} terms. Given a formula and a data frame containing the data, the
\texttt{ngme()} function assembles the individual components into a coherent statistical model and
performs likelihood based estimation under user-specified control options.
The fitted object can be used to perform prediction discussed in Section~\ref{sec:model-prediction} using a \texttt{predict} function.

For a complete introduction to the functionality and options of the package, see the package homepage \url{https://davidbolin.github.io/ngme2/}.

\section{Applications} \label{chp:applications}

This section presents four applications using LLnGMs of increasing complexity; a time series model, a longitudinal model, a spatial model, and a multivariate spatio-temporal model. 
For model comparison in each application, we use  
mean absolute errors (MAE), mean squared errors (MSE), the (negative) continuous ranked probability score (CRPS) \citep{crps}, and the (negative) scaled continuous ranked probability score (sCRPS)  \citep{bolin2023}.
MAE and MSE are standard measures of model accuracy, while CRPS and sCRPS are probabilistic scores that quantify the quality of the predictive distribution. 
All analyses were performed using \texttt{ngme2} on an M1 Mac Max 2021 with 32GB RAM, and all estimations were run until convergence by the criteria in Section~\ref{sec:convergence-diagnostics}. 

\change{
For NIG models, we assign a penalizing complexity prior to the 
mixing parameter $\nu$, specifically an Inverse-Exponential prior $1/\nu \sim \mathrm{Exp}(\lambda)$ \citep{cabral2023}. 
This penalizes deviations away from a Gaussian baseline model which is recovered as $\nu \to \infty$. To ensure a 
weakly informative setting, we calibrate the rate parameter as 
$\lambda = \ln(2) / \mathrm{median}(\mathbf{h})$. This choice implies that the prior median of 
$\nu$ matches the reciprocal of the representative scale of $\mathbf{h}$, 
resulting in a prior-median coefficient of variation $\mathrm{CV}(V) \approx 1$. 
This allows for substantial departures from Gaussianity 
while maintaining model stability.
An alternative calibration which would need to be done separately for each specific model is to choose $\lambda$ by considering the ratio of the probability of large marginal events in the 
non-Gaussian case relative to its Gaussian counterpart \citep[see][Appendix~E]{cabral2023}. 
For all other parameters, the \texttt{ngme2} analysis employs default 
weakly informative priors $\mathcal{N}(0, 10)$ on the parameters transformed to the unconstrained space.}
Code to replicate all results is available at \url{https://github.com/MuggleJinx/LLNGM-code}.

\subsection{The Grasshopper population data}

We analyze the grasshopper abundance data from Montana, USA, covering the period from 1948 to 1990 (with the years 1949, 1950, 1977, and 1982 missing), which is a classic example of wildlife time-series data with ecological fluctuations and population dynamics. 
The grasshopper dataset was originally reported by \citet{kemp1993}.
Figure~\ref{fig:grasshopper-data} illustrates the population trends over the observed period. 

\begin{figure}[t]
  \centering
  \subfloat[\label{fig:grasshopper-data}]{\includegraphics[width=0.48\textwidth]{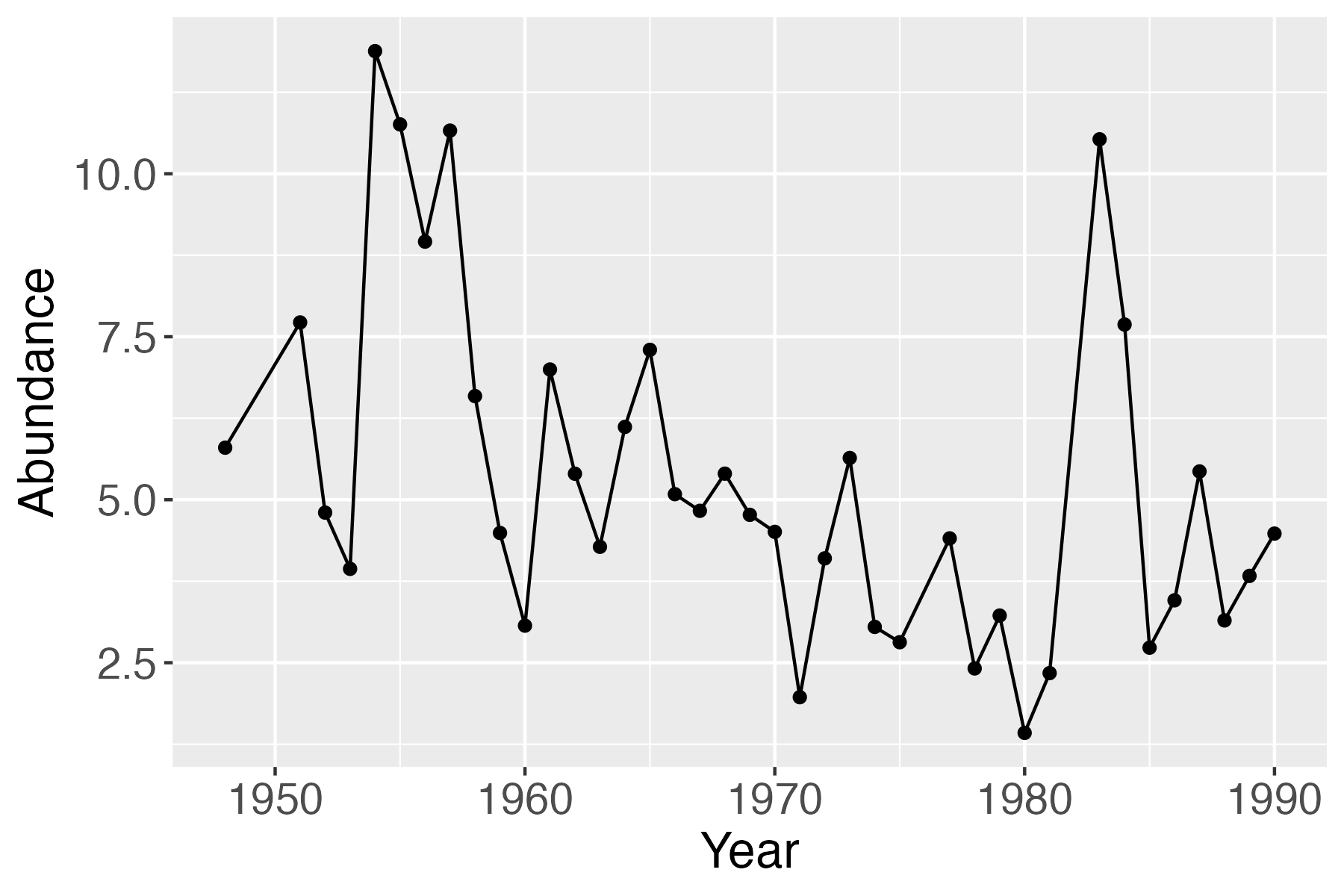}}
  \hfill
  \subfloat[\label{fig:grasshopper-predictions}]{\includegraphics[width=0.48\textwidth]{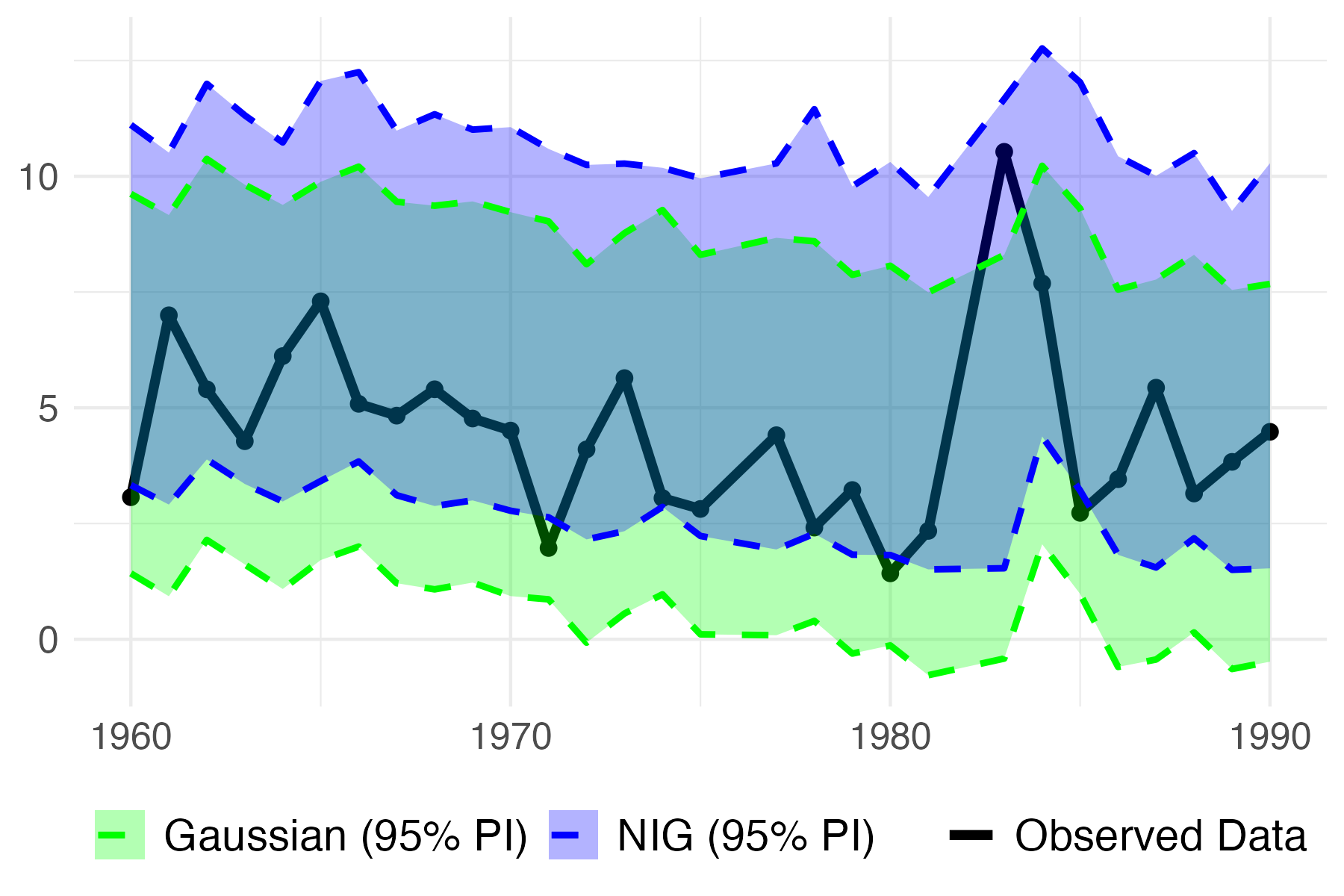}}
  \caption{Grasshopper Population Analysis: (a) Time Series Data and (b) Rolling Window 95\% Prediction Intervals from Gaussian and NIG AR(1) Models}
  \label{fig:grasshopper}
\end{figure}

Let $Y(t)$ denote the observed grasshopper abundance at year $t$. 
To account for long-term population trends, we include a slope fixed effect and specify the observation model as $Y(t) = \beta_0 + \beta_1 t_{\text{scaled}} + W(t) + \epsilon^{\mathbf{Y}}(t)$, where $t_{\text{scaled}}$ is the scaled year, $\epsilon^{\mathbf{Y}}(t) \sim \mathcal{N}(0, \sigma_\epsilon^2)$, and $W(t) = \rho W(t-1) + \epsilon^{\mathbf{W}}(t)$ is an AR(1) process defined on a yearly spaced mesh.
We consider two different specifications for the latent process innovations:
Gaussian AR(1): $\epsilon^{\mathbf{W}}(t) \sim \mathcal{N}(0, \sigma^2)$; 
NIG AR(1): $\epsilon^{\mathbf{W}}(t) \sim \text{NIG}(\mu, \sigma, \nu)$.
The posterior mean and 95\% credible intervals are reported in Table \ref{tab:grasshopper-estimates}. 

The inclusion of the linear trend explicitly accounts for the long-term decline in the population (with $\hat{\beta}_1$ estimated around $-1.05$ to $-0.86$). Both specifications estimate a similar autoregressive dependence ($\hat{\rho} \approx 0.31$ to $0.37$). While the Gaussian AR(1) model forces the process innovations to be symmetric, the NIG AR(1) model reveals significant positive skewness ($\hat{\mu}=2.41$) and heavy-tailed behavior ($\hat{\nu}=1.33$), suggesting that ecological shocks are structurally asymmetric.

\begin{table}[t]
\centering
\caption{Parameter Estimates for Grasshopper Population Models}
\label{tab:grasshopper-estimates}
\begin{tabular}{llcc}
\hline
Parameter type & Parameter & Gaussian Model & NIG Model \\
\hline
Fixed Effects & Intercept $\beta_0$ & 5.35 (4.69, 6.14) & 5.20 (4.64, 5.63) \\
& Slope $\beta_1$ & $-1.05$ ($-1.26$, $-0.88$) & $-0.86$ ($-1.01$, $-0.79$) \\
\hline
Latent Process & Autocorrelation $\rho$ & 0.31 (0.11, 0.53) & 0.37 (0.15, 0.54) \\
& Scale $\sigma$ & 2.13 (1.69, 2.66) & 0.47 (0.21, 0.93) \\
& Skewness $\mu$ & -- & 2.41 (1.75, 3.10) \\
& Shape $\nu$ & -- & 1.33 (0.80, 2.57) \\
\hline
Measurement & Scale $\sigma_\epsilon$ & 0.17 (0.08, 0.33) & 0.84 (0.53, 1.95) \\
\hline
\end{tabular}
\end{table}

Cross-validation and prediction intervals are conducted using a rolling window approach. Using the estimated model, 
we predict data for the current year based on the data from the preceding 10 years.
As shown in Table~\ref{tab:grasshopper-cv-results}, the NIG model provides competitive absolute point predictions (lower MAE, despite a slightly higher MSE) and superior probabilistic calibration (lower CRPS and sCRPS). Thus, the overall scores still favor the NIG model. 

Beyond aggregate metrics, the true advantage of the NIG model lies in its ability to capture asymmetric uncertainty, providing insights into population dynamics that are unavailable from symmetric Gaussian predictions. This asymmetry is obvious when comparing the 95\% predictive intervals in Figure \ref{fig:grasshopper-predictions}. 
Driven by the strong positive skewness ($\hat{\mu}=2.41$) of the latent process, the predictive intervals from the NIG model are systematically shifted upward, with the observed data frequently tracing the lower bounds. This produces higher predictive means and substantially wider upper bounds, and indicates that the NIG model realistically quantifies the asymmetric risk of sudden population outbreaks, whereas the Gaussian model produces symmetric intervals that fail to capture the underlying ecological reality.

\begin{table}[t]
\centering
\caption{Cross-Validation Results for Grasshopper Population Models}
\label{tab:grasshopper-cv-results}
\begin{tabular}{lrrrr}
\hline
 & MAE & MSE & -CRPS & -sCRPS \\
\hline
Gaussian & 1.415 & \textbf{3.601} & 1.032 & 1.368 \\
NIG & \textbf{1.382} & 3.604 & \textbf{0.964} & \textbf{1.337} \\
\hline
\end{tabular}
\end{table}

\subsection{Longitudinal data analysis}

Longitudinal studies provide insights into disease progression but present challenges including different observation times, patient heterogeneity, and irregular sample paths including sudden jumps and spikes.
We analyze clinical kidney function data from the northern English city of Salford, focusing on patients in high-risk groups for chronic kidney disease \citep{diggle2015}. The dataset contains 392,870 measurements from 22,910 primary care patients, with follow-up periods ranging from baseline only to 10.0 years. 
The cohort includes 11,833 males and 11,077 females, with baseline ages spanning 13.7 to 102.1 years (mean: 65.4 years).

\begin{figure}[t]
  \centering
  \includegraphics[width=0.6\textwidth]{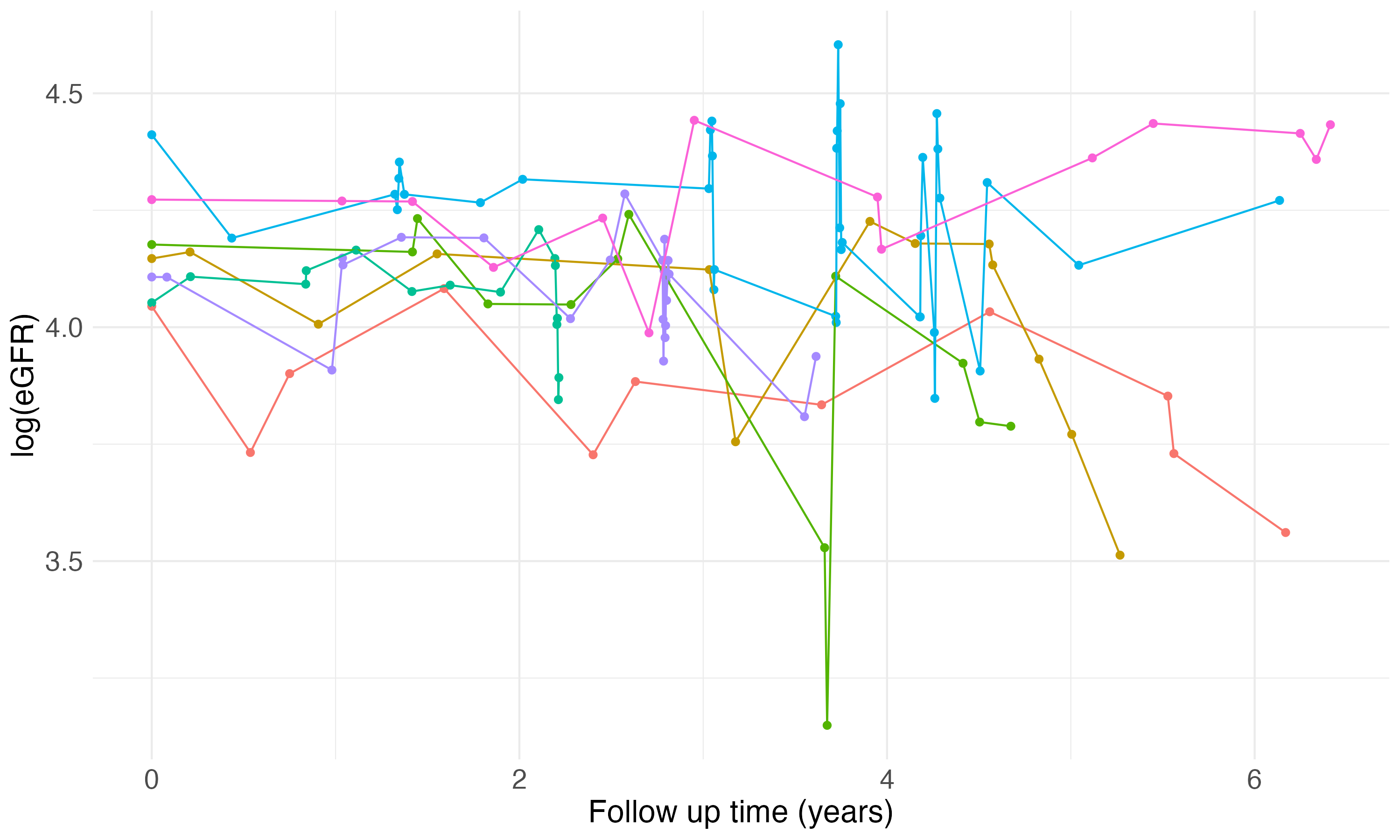}
  \caption{Longitudinal eGFR Trajectories for 7 Randomly Selected Patients Showing Individual Patterns of Kidney Function Change over Follow-Up Time}
  \label{fig:srft-data-sample}
\end{figure}

The outcome variable is eGFR (estimated glomerular filtration rate, in millilitres per minute per 1.73 metres squared of body surface area), which is a proxy measurement for the patient's renal function \citep{levey1999}. 
%
Figure~\ref{fig:srft-data-sample} shows typical longitudinal eGFR trajectories for five randomly selected patients. The trajectories exhibit considerable variability in both the frequency and magnitude of measurements, reflecting the real-world nature of clinical data collection. The irregular follow-up patterns arise from routine clinical practice, where patients are monitored based on medical necessity rather than fixed schedules. 

We analyzed a random selection of 500 patients with more than 10 measurements each, containing a total of $13196$ observations.
We consider the following model specification for $i$-th patient and $j$-th measurement:
\begin{equation}
\log(\text{eGFR}_{ij}) = \beta_0 + \beta_1 \text{sex}_i + \beta_2 \text{bage}_i + \gamma_i + W_i(t_{ij}) + \epsilon^{\mathbf{Y}}_{ij}
\end{equation}
where $\gamma_i \sim N(0, \sigma_\gamma^2)$ are patient-specific random intercepts, $W_i(t)$ is a Matérn SPDE model ($\alpha = 2$) along follow-up time, and $\epsilon^{\mathbf{Y}}_{ij} \sim N(0, \sigma_\epsilon^2)$ is measurement error. 
We consider two model specifications for the Matérn process noise; Gaussian and NIG. 

\begin{table}[t]
\centering
\caption{Parameter Estimates for Kidney Function Models (500 patients)}
\label{tab:srft-estimates}
\begin{adjustbox}{max width=0.90\textwidth}
\begin{tabular}{llcc}
\hline
Parameter type & Parameter & Gaussian Model & NIG Model \\
\hline
Fixed Effects & Intercept & 4.07 (4.00, 4.14) & 4.11 (4.05, 4.18) \\
& Sex & $-0.070$ ($-0.076$, $-0.063$) & $-0.083$ ($-0.089$, $-0.079$) \\
& Baseline age & $-0.009$ ($-0.010$, $-0.009$) & $-0.011$ ($-0.011$, $-0.010$) \\
\hline
Random Effects &  Scale $\sigma_\gamma$ & 0.025 (0.020, 0.031) & 0.016 (0.012, 0.022) \\
\hline
Process & Range $\kappa$ & 1.17 (1.10, 1.26) & 1.15 (1.07, 1.24) \\
& Scale $\sigma$ & 0.78 (0.71, 0.87) & 1.26 (1.13, 1.39) \\
& Skewness $\mu$ & -- & $-0.05$ ($-0.10$, $0.00$) \\
& Shape $\nu$ & -- & 0.075 (0.061, 0.092) \\
\hline
Measurement & Scale $\sigma_\epsilon$ & 0.211 (0.207, 0.215) & 0.179 (0.175, 0.182) \\
\hline
\end{tabular}
\end{adjustbox}
\end{table}

\begin{table}[t]
\centering
\caption{Cross-Validation Results for Kidney Function Models (500 patients)}
\label{tab:srft-cv-results}
\begin{tabular}{lrrrr}
\hline
 & MAE & MSE & -CRPS & -sCRPS \\
\hline
Gaussian & 0.142 & 0.047 & 0.113 & 0.272 \\
NIG & \textbf{0.131} & \textbf{0.038} & \textbf{0.100} & \textbf{0.206} \\
\hline
\end{tabular}
\end{table}

Parameter estimates with 95\% credible intervals and cross-validation results are reported in Table \ref{tab:srft-estimates} and Table \ref{tab:srft-cv-results}, respectively.
The NIG model achieves better results in all metrics, demonstrating the advantages of a non-Gaussian approach for complex longitudinal medical data, which is in line with the results in \cite{asar2020}.

\subsection{Climate Reanalysis Data}
We now explore a spatial application using climate reanalysis data, where we compare Gaussian and non-Gaussian models, as well as stationary and non-stationary models. The data is derived from the Experimental Climate Prediction Center Regional Spectral Model (ECPC-RSM), initially developed for the North American Regional Climate Change Assessment Program (NARCCAP) using NCEP/DOE Reanalysis \citep{mearns2009,mearns2014}.
We follow a similar analysis as in \cite{bolin2020b}, focusing on the average summer precipitation over the conterminous United States. To fit the models, we use data from the first 10 years (1979--1988) as independent temporal replicates. To evaluate predictive performance, we use the data from the year 1989 as a hold-out test set, which contains 4,112 grid cells.
Following standard practice, we apply a cube root transformation to the data to improve Gaussianity and subtract temporal means, yielding the model
$Y_{ij} = W_j(\mathbf{s}_i) + \epsilon^{\mathbf{Y}}_{ij}$,
where $Y_{ij}$ is the observation at location $\mathbf{s}_i$ and year $j$, $\epsilon^{\mathbf{Y}}_{ij} \sim \mathcal{N}(0, \sigma^2)$ are measurement errors and $\{W_j(\mathbf{s})\}_j$ are independent spatial field realizations.
Figure~\ref{fig:precip-data} displays the spatial distribution of mean summer precipitation across the contiguous United States for 1979. 

Given the large spatial domain, we expect non-stationary behavior due to topography, latitude effects, and regional climate differences. We model $W$ as a Mat\'ern SPDE field and compare four specifications with either Gaussian or NIG driving noise, and either stationary or non-stationary parameters. The non-stationary models allow spatial variation in Matérn parameters through an altitude covariate and Fourier basis functions as described in \cite{bolin2020b}: 
\begin{equation*}
\log \kappa(\mathbf{s}) = \kappa_0 + \kappa_a \psi_a(\mathbf{s}) + \sum_{i,j=1}^2 \sum_{k,\ell=1}^2 \kappa_{ij}^{k\ell} \psi_i^k(\tilde{s}_1) \psi_j^\ell(\tilde{s}_2),
\label{eq:nonstat-kappa}
\end{equation*}
with analogous formulation for $\log \sigma(\mathbf{s})$, and for the NIG parameters $\mu(\mathbf{s})$ and $\log \nu(\mathbf{s})$. This adds 18 parameters for each of the Matérn parameters to capture topographic effects and large-scale spatial trends.

\begin{figure}[t]
\centering
\includegraphics[width=0.6\textwidth]{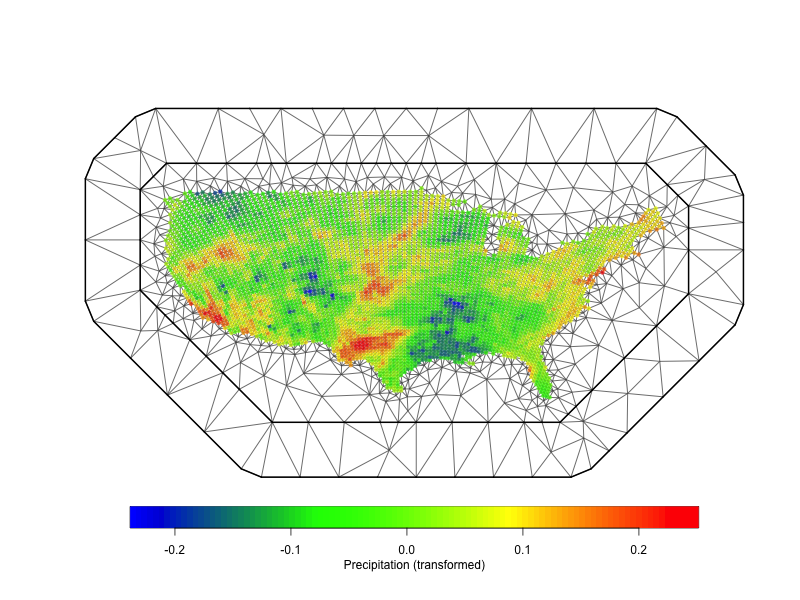}
\caption{Transformed Precipitation Data for the United States with the Triangulation Mesh Used for Spatial Modeling.}
\label{fig:precip-data}
\end{figure}

\begin{table}[t]
\centering
\caption{Prediction Scores of Climate Precipitation Models on 1989 Test Data}
\begin{tabular}{lrrrrr}
\hline
\multirow{2}{*}{Model} & \multirow{2}{*}{Fit time(s)} & \multicolumn{4}{c}{Prediction Accuracy} \\
\cline{3-6}
 & & MAE & MSE & -CRPS & -sCRPS \\
\hline
Gaussian stationary & 89 & 0.01920 & \textbf{0.00065} & 0.01396 & -0.78057 \\
Gaussian non-stationary & 249 & \textbf{0.01917} & \textbf{0.00065} & \textbf{0.01367} & \textbf{-0.81082} \\
NIG stationary & 75 & 0.01921 & 0.00066 & 0.01389 & -0.78689 \\
NIG non-stationary & 206 & 0.01951 & 0.00067 & 0.01387 & -0.80018 \\
\hline
\end{tabular}
\label{tab:cv-scores}
\end{table}

The models are evaluated on the basis of their spatial prediction performance on the 1989 test year. The results in Table~\ref{tab:cv-scores} show that the Gaussian non-stationary model outperforms the others across all evaluation metrics. Furthermore, the benefits of incorporating non-stationarity are evident regardless of the noise assumption: the NIG non-stationary model achieves better CRPS and sCRPS than both stationary models. This underscores the importance of accounting for spatial heterogeneity in precipitation modeling. Conversely, the non-Gaussian specifications provide no additional overall benefit, suggesting that the initial cube root transformation brings the data sufficiently close to Gaussianity. Thus, for this data, non-stationarity is the main driver of predictive improvement. 

\subsection{Space-time bivariate model for the Wind Data}

As a final application, we illustrate the LLnGM framework to bivariate space-time data, specifically using wind vector fields from the NCEP/NCAR Reanalysis Project \cite{NCEPNCAR}. Our study is centered on the Balkan Peninsula (10°E-36°E, 34°N-48°N) over a 60-hour period in 2015, divided into 10 six-hour intervals, resulting in 1,078 space-time observations of the zonal (u-wind) and meridional (v-wind) components.
Figure~\ref{fig:wind-data} shows the wind data over the Balkan Peninsula region.

Let $\mathbf{Y}(\mathbf{s},t) = (Y_u(\mathbf{s},t), Y_v(\mathbf{s},t))^\top$ denote the bivariate wind vector at spatial location $\mathbf{s} = (s_1, s_2)$ and time $t$, where $s_1$ and $s_2$ represent longitude and latitude, respectively. We consider the additive decomposition $\mathbf{Y}(\mathbf{s},t) = \mathbf{X}(\mathbf{s})\boldsymbol{\beta} + \mathbf{W}(\mathbf{s},t) + \boldsymbol{\epsilon}^{\mathbf{Y}}(\mathbf{s},t)$. Here, the mean term $\mathbf{X}(\mathbf{s})\boldsymbol{\beta}$ captures spatial trends via fixed effects, where $\mathbf{X}(\mathbf{s})$ is a block-diagonal design matrix of spatial covariates and $\boldsymbol{\beta}$ is the corresponding vector of regression coefficients, such that the expected values for the $u$- and $v$-components are $\beta_{u,0} + \beta_{u,1}s_1$ and $\beta_{v,0} + \beta_{v,2}s_2$, respectively. The term $\mathbf{W}(\mathbf{s},t)$ is a zero-mean bivariate latent space-time field, modeled using the multivariate AR construction with operator matrix \eqref{eq:tensor-multivariate-construction} with Gaussian or NIG noise,  to represent the zonal ($u$-wind) and meridional ($v$-wind) components. This model integrates temporal, spatial, and cross-correlations through a separable space-time bivariate structure. Finally, $\boldsymbol{\epsilon}^{\mathbf{Y}}(\mathbf{s},t)$ denotes an independent measurement-noise term.

\begin{figure}[t]
  \centering
  \includegraphics[width=0.7\textwidth]{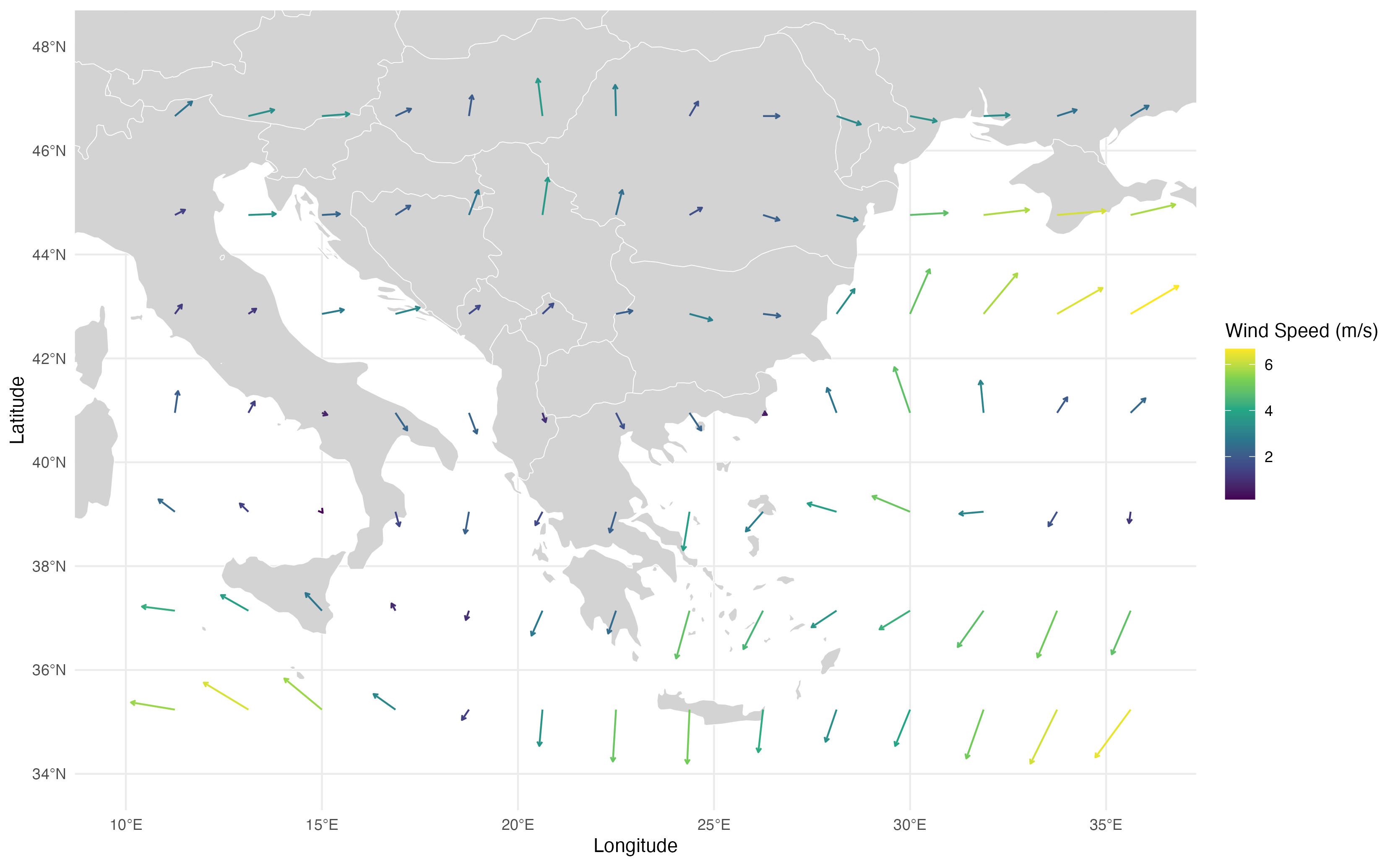}
  \caption{Wind Data over the Balkan Peninsula Region. Arrows Indicate Wind Direction, While Colors Represent Wind Speed Magnitude.}
  \label{fig:wind-data}
\end{figure}
\begin{table}[t]
\centering
\caption{Mean Cross-Validation Scores for Wind Data Models}
\begin{tabular}{lcrrrr}
\hline
\multirow{2}{*}{Model} & \multirow{2}{*}{Fit Time (min)} & \multicolumn{4}{c}{Prediction Accuracy} \\
\cline{3-6}
 &  & MAE & MSE & -CRPS & -sCRPS \\
\hline
Gaussian & 1.48 & 0.832 & \textbf{1.046} & 0.636 & 1.138 \\
NIG & 7.28 & \textbf{0.809} & 1.113 & \textbf{0.621} & \textbf{1.116} \\
\hline
\end{tabular}
\label{tab:cv-scores-wind}
\end{table}

We again compare the models through 10-fold cross-validation. The NIG model outperforms the Gaussian model in absolute point prediction (lower MAE) and probabilistic calibration (lower CRPS and sCRPS). Although the Gaussian model yields a slightly lower MSE, the overall scores favor the NIG model (Table~\ref{tab:cv-scores-wind}). This highlights the NIG model's ability to capture non-Gaussian features of wind, such as asymmetry, heavy tails, and intermittent behavior, as well as complex dependencies between zonal and meridional components due to atmospheric dynamics.


\section{Comparison with other methods and software}\label{sec:comparison}

In this section, we compare our LLnGM framework with the main alternative approaches for fitting non-Gaussian latent models.
As mentioned in the introduction, common Bayesian modeling packages such as \texttt{R-INLA} \citep{r-inla} and \texttt{brms} \citep{brms} do not support LLnGMs.  
The recent package \texttt{ngvb} \citep{cabral2024} provides algorithms for fitting latent non-Gaussian models. 
However, the package is not maintained, so we instead compared with the more general software \texttt{Stan} \citep{stan} and \texttt{TMB} \citep{kristensen2016tmb}.
%

We fit an AR(1) model with NIG-distributed innovations to 500 simulated observations from the model, subject to Gaussian measurement error. We compare three different approaches:
\texttt{ngme2}: Rao-Blackwellized SGD optimization 
followed by posterior exploration using SGLD to draw 2,000 samples.
\texttt{Stan}: Full MCMC with 1,000 warm-up iterations and 2,000 sampling iterations across four parallel chains.
\texttt{TMBStan}: MCMC with identical iteration settings (1,000 warm-up, 2,000 sampling) and four parallel chains, utilizing a Laplace approximation to integrate out the latent field and mixing variables.
\change{All three methods (\texttt{ngme2}, \texttt{Stan}, and \texttt{TMBStan}) use the same weakly informative priors on shared parameters.}
The code is available at \url{https://github.com/MuggleJinx/LLNGM-code}.

\begin{table}[t]
\centering
\caption{Comparison of True Parameter Values, Posterior Means, and 95\% Credible Intervals for Each Method}
\label{tab:parameter-comparison}
\begin{tabular}{lcccc}
\hline
Parameter & True Value & \texttt{ngme2} & STAN & TMBSTAN \\
\hline
$\sigma_{\epsilon}$ & 1.00 & 1.00 (0.85, 1.19) & 1.04 (0.75, 1.27) & 1.39 (1.23, 1.58) \\
$\rho$ & 0.80 & 0.80 (0.77, 0.82) & 0.80 (0.77, 0.82) & 0.80 (0.78, 0.82) \\
$\mu$ & 3.00 & 3.03 (2.73, 3.34) & 3.14 (2.47, 3.92) & 3.30 (3.12, 3.50) \\
$\sigma$ & 2.00 & 1.72 (1.15, 2.21) & 1.26 (0.37, 2.59) & 0.10 (0.01, 0.48) \\
$\nu$ & 0.40 & 0.36 (0.24, 0.48) & 0.41 (0.22, 0.73) & 0.43 (0.36, 0.53) \\
\hline
Noise KLD & -- & 0.011 & 0.075 & 2.026 \\
Elapsed time (s) & -- & 20.0 & 37.7 & 359.7 \\
\hline 
\end{tabular}

\begin{minipage}{0.90\textwidth}
\footnotesize
\textit{Note: \texttt{ngme2} directly reports $\rho$, $\mu$, $\sigma$, $\nu$, and $\sigma_{\epsilon}$. STAN and TMBSTAN report transformed parameters $\rho_{\text{un}} = \log((1+\rho)/(1-\rho))$ (unconstrained AR(1) parameter), $\log(\sigma)$, $\log(\nu)$, and $\log(\sigma_{\epsilon})$, which have been back-transformed for comparison. The noise KLD is calculated by comparing the true and estimated noise distributions.}
\end{minipage}
\end{table}

The parameter estimates with 95\% credible intervals are shown in Table \ref{tab:parameter-comparison}, and the full results are provided in Appendix~\ref{app:comparison}.
Across the three approaches, \texttt{ngme2} produced the most stable and computationally efficient fit for this latent NIG-AR(1) model, completing the estimation in approximately 8.4 seconds with proper convergence diagnostics and sampling process in 11.6 seconds.

The \texttt{Stan} fit required 37.7 seconds and exhibited clear pathologies: 10 divergent transitions, low Bayesian Fraction of Missing Information (BFMI) across 4 chains, and R-hat values up to 2.11 indicating poor chain mixing. Heavy-tailed latent innovations create a funnel geometry between the parameters ($\sigma$, $\nu$) and the latent states, forcing NUTS to take extremely small steps to navigate this challenging posterior landscape. This results in severely collapsed effective sample sizes (as low as 4 for $\log(\sigma)$) and biased parameter estimates, particularly for $\hat{\sigma}$, which was severely underestimated (1.26 vs. true value 2.00) because the sampler fails to adequately explore the tails and low-probability regions. 

\texttt{TMBStan} uses a second-order Laplace approximation to integrate out the latent field and the mixing variables, followed by HMC on the remaining parameters. While this approach resolved the sampling pathologies and yielded perfect chain mixing (R-hat of 1.00), it introduced severe drawbacks. First, evaluating the dense Hessian for the Laplace approximation at every leapfrog step caused the computational wall-time to increase to nearly 360 seconds. Second, the Gaussian assumption underpinning the Laplace approximation around the mode entirely fails to capture the heavy-tailed, asymmetric true posterior of the NIG model. Consequently, TMBStan severely biased the results by collapsing the latent scale parameter ($\hat{\sigma} \approx 0.10$) and erroneously absorbing the excess variance into the Gaussian measurement error ($\hat{\sigma}_{\epsilon} \approx 1.39$).

The estimation of $\rho$ remained robust across all methods. However, to directly compare how well each method captures the distribution of the driving noise, we use the Kullback-Leibler divergence (KLD) between the true and estimated noise distributions. \texttt{ngme2} achieves the lowest KLD of 0.011, successfully recovering the non-Gaussian structure of the driving noise. \texttt{Stan} had a KLD of 0.075, while \texttt{TMBStan} completely misses the true distribution with a KLD of 2.026. 

In summary, 
\texttt{ngme2} performs better both in computational efficiency and in estimation accuracy compared to the alternatives. 
This highlights the need for specialized frameworks like \texttt{ngme2} when latent non-Gaussianity is essential.

\section{Discussion}\label{chp:discussion}

We have introduced the LLnGM framework as a unified approach that extends latent Gaussian models to accommodate non-Gaussian behaviors while maintaining computational efficiency and scalability. The flexibility of the methodology was illustrated in four applications, including novel spatial and spatio-temporal non-Gaussian models.
The framework leverages mean-variance mixture representations through the GH family of distributions, enabling flexible modeling of complex dependencies including temporal, spatial, multivariate, and spatio-temporal structures. Through the exploitation of sparse precision matrices and stochastic gradient Langevin dynamics (SGLD) methodology enhanced with Rao-Blackwellization, the approach achieves computational tractability for high-dimensional problems. The accompanying \texttt{ngme2} R package provides a user-friendly implementation, addressing the current gap in available software for latent non-Gaussian modeling. 

The framework has maximum likelihood estimation (MLE) as a special case of the Bayesian formulation. The transition from MAP estimation to MLE requires only a straightforward modification: rather than optimizing the log-posterior $\log \pi(\boldsymbol{\theta} | \mathbf{Y}) = \log \pi(\mathbf{Y} | \boldsymbol{\theta}) + \log \pi(\boldsymbol{\theta})$, we maximize the log-likelihood $\ell(\boldsymbol{\theta} | \mathbf{Y}) = \log \pi(\mathbf{Y} | \boldsymbol{\theta})$. Consequently, the existing SGD methodology remains applicable by simply removing the gradient of the log-prior from the objective function.



The LLnGM framework establishes a robust foundation for several promising research directions. A natural extension is to incorporate additional model classes into the framework. For example, extending the rational approximation of fractional SPDEs \citep{bolin2020b} to non-Gaussian settings is a natural next step.
The framework could also be extended to handle non-Gaussian random fields defined on metric graphs \citep{bolin2023a}, to address applications where the underlying spatial domain has a graph structure. 
Another important direction involves generalizing the observation process by allowing parameterization of the observation matrix $\mathbf{A}$, which would enable more flexible modeling of how the latent process relates to the observed data. This could accommodate scenarios with time-varying observation patterns, measurement uncertainties, or complex sampling designs, thereby broadening the applicability of the framework to a wider range of real-world problems where the observation mechanism itself is of scientific interest.

\section*{Appendix}
\appendix
\providecommand\K{\mathbf{K}}
\providecommand\D{\mathbf{D}}
\providecommand\C{\mathbf{C}}
\providecommand\G{\mathbf{G}}
\providecommand\h{\mathbf{h}}
\providecommand\V{\mathbf{V}}
\providecommand\W{\mathbf{W}}
\providecommand\X{\mathbf{X}}
\providecommand\A{\mathbf{A}}
\providecommand\B{\mathbf{B}}
\providecommand\M{\mathbf{M}}
\providecommand\Y{\mathbf{Y}}
\providecommand\bs{\boldsymbol}

\section{Operators for Different Model Classes}\label{app:operators}

Within the LLnGM framework, the operator $\mathbf{K}$ serves as the primary mechanism for specifying the latent dependence structure.
Different modeling assumptions are realized through different choices of $\mathbf{K}$.
To emphasize modularity, we organize examples by how the latent dependence operator is constructed and composed in the rest of this section. 
In Table~\ref{tab:operator-matrices}, we summarize the operator matrices for used model classes, illustrating the diversity of models which can be represented within a unified operator-based formulation.

\subsection{Temporal operators}

A large class of time series models can be written as linear difference (or recursion) equations of the form
\begin{equation}
\mathbf{K}_t(\boldsymbol{\theta}_t)\,\mathbf{W} = \boldsymbol{\epsilon},
\label{eq:many-temporal-Kw}
\end{equation}
where $\mathbf{W}=(w_1,\ldots,w_T)^\top$ is a time series vector and $\mathbf{K}_t(\boldsymbol{\theta}_t)$ is a sparse (typically banded) matrix.

\emph{Autoregressive models.}
For an AR($p$) process, $\phi(B)w_t=\epsilon_t$ with $\phi(B)=1-\phi_1B-\cdots-\phi_pB^p$, the corresponding
operator matrix $\mathbf{K}_t(\boldsymbol{\theta}_t)$ is banded with bandwidth $p$, reflecting the Markov property of order $p$. For example, for an AR(1) process $w_t - \phi w_{t-1} = \epsilon_t$ with $t=1,\dots,T$, the operator matrix is given by
\begin{equation}
\mathbf{K}_{\text{AR1}}(\phi) = 
\begin{pmatrix} 
\sqrt{1-\phi^2} & 0 & \cdots & 0 \\
-\phi & 1 & \cdots & 0 \\
\vdots & \ddots & \ddots & \vdots \\
0 & \cdots & -\phi & 1
\end{pmatrix}.
\label{eq:AR1-matrix}
\end{equation}
The term $\sqrt{1-\phi^2}$ in the first row ensures that the variance of $w_1$ matches the stationary variance of the process. 
\change{
Since the AR process is indexed by integers, we have $\mathbf{h} = \mathbf{1}$.
}

\emph{Random-walk priors.}
Random walk models correspond to difference operators. \change{Given a set of ordered locations $x_1 < x_2 < \dots < x_n$}, for a first-order random walk RW(1), 
\change{$\Delta w_i := w_i - w_{i-1} = \epsilon_i$}, the operator matrix \change{$\mathbf{K}_{\text{RW1}} \in \mathbb{R}^{(n-1) \times n}$} is the first-difference matrix:
\begin{equation}
\change{\mathbf{K}_{\text{RW1}}} = 
\begin{pmatrix} 
-1 & 1 & 0 & \cdots & 0 \\
0 & -1 & 1 & \cdots & 0 \\
\vdots & \vdots & \ddots & \ddots & \vdots \\
0 & 0 & \cdots & -1 & 1
\end{pmatrix}.
\end{equation}
For a second-order random walk RW(2), \change{$\Delta^2 w_i := w_i - 2w_{i-1} + w_{i-2} = \epsilon_i$}, the operator matrix \change{$\mathbf{K}_{\text{RW2}} \in \mathbb{R}^{(n-2) \times n}$} is the second-difference matrix:
\begin{equation}
\change{\mathbf{K}_{\text{RW2}}} = 
\begin{pmatrix} 
1 & -2 & 1 & 0 & \cdots & 0 \\
0 & 1 & -2 & 1 & \cdots & 0 \\
\vdots & \vdots & \ddots & \ddots & \ddots & \vdots \\
0 & 0 & \cdots & 1 & -2 & 1
\end{pmatrix}.
\end{equation}
These models are widely used as flexible smoothness priors, and their sparse banded operator matrices make them computationally attractive for large-scale latent modeling. Furthermore, the innovations \change{$\epsilon_i$} are assumed to have variance \change{$\sigma^2 h_i$, where $h_i = x_{i} - x_{i-1}$} represents the mesh weights defining the spacing between consecutive locations. For regular grids, \change{$h_i = h$} is a constant distance, yielding standard random walk behavior. For irregular grids, varying \change{$h_i$} allows the innovation variance to adapt to local sampling densities, ensuring scale invariance and proper weighting across different spatial or temporal gaps.

\emph{Ornstein--Uhlenbeck (OU) models.}
The OU process is a continuous-time analogue of an AR(1) process, defined by the stochastic differential equation (SDE) $dX_t = -\theta X_t dt + \sigma dB_t$, where $B_t$ denotes standard Brownian motion. In the Gaussian setting, the process admits an exact discrete-time solution over a temporal mesh $t_1 < t_2 < \dots < t_n$ \change{with potentially non-uniform mesh weights (step sizes) $h_i = t_{i+1} - t_i$. This solution follows the recursion $w_{i+1} = \rho_i w_i + \epsilon_i$, where the autoregressive coefficient is $\rho_i = \exp(-\theta h_i)$.} This yields a banded operator matrix $\mathbf{K}(\theta)$ of the form:
\begin{equation}
\mathbf{K}(\theta) = 
\begin{pmatrix} 
\sqrt{1-\rho_1^2} & 0 & 0 & \cdots & 0 \\
-\rho_1 & 1 & 0 & \cdots & 0 \\
0 & -\rho_2 & 1 & \cdots & 0 \\
\vdots & \vdots & \ddots & \ddots & \vdots \\
0 & 0 & \cdots & -\rho_{n-1} & 1
\end{pmatrix},
\label{eq:OU-matrix}
\end{equation}
where the entry $\sqrt{1-\rho_1^2}$ ensures the variance of $w_1$ aligns with the stationary variance. While $\change{\mathbf{K}(\theta)}$ is derived from the Gaussian-driven SDE, we utilize this structure as a robust operator-based approximation for non-Gaussian OU models. \change{Importantly, the non-Gaussian innovations $\epsilon_i$ are intrinsically scaled by the mesh weights $h_i$. By explicitly incorporating $h_i$ into both the temporal dependencies ($\rho_i$) and the innovation variance, the framework effectively captures non-Gaussian dynamics and maintains consistent statistical properties across irregular observation intervals.}

\subsection{Spatial operators}
A canonical spatial example is the Mat\'ern SPDE construction, where the latent field $w(\cdot)$ is defined
as the (weak) solution to
\begin{equation}
(\kappa^2 - \Delta)^{\alpha/2} \, w(s) = \mathcal{W}(s),
\label{eq:many-matern-spde}
\end{equation}
with $\mathcal{W}$ denoting the spatial white noise. For the case $\alpha=2$, after finite element discretization using \change{hat} functions $\{\psi_i\}_{i=1}^n$ on a mesh, \eqref{eq:many-matern-spde} yields a sparse linear system
$\mathbf{K}(\boldsymbol{\theta})\,\mathbf{W} = \boldsymbol{\epsilon}$,
where $\mathbf{W}\in\mathbb{R}^{n}$ collects nodal basis coefficients and $\mathbf{K}(\boldsymbol{\theta})$ is the sparse operator matrix defined as:
$\mathbf{K}(\boldsymbol{\theta}) = \kappa^2 \mathbf{C} + \mathbf{G}$.
Here, $\mathbf{C}$ is the \emph{mass matrix} with entries $C_{ij} = \langle \psi_i, \psi_j \rangle$, which represents the overlap of basis functions (often diagonalized via a lumping approximation to preserve sparsity and simplify the innovation structure). \change{Specifically, this lumping yields the spatial mesh weights $h_i = \sum_j C_{ij}$, which represent the local area associated with node $i$. Analogous to the 1D models, these weights $h_i$ govern the variance scaling of the innovations $\boldsymbol{\epsilon}$, allowing the framework to consistently handle non-Gaussian noise over irregular spatial triangulations.} The matrix $\mathbf{G}$ is the \emph{stiffness matrix} with entries $G_{ij} = \langle \nabla \psi_i, \nabla \psi_j \rangle$, representing the discretized Laplacian operator $\Delta$. This construction links the physical parameters of the SPDE directly to a sparse matrix representation, providing a scalable route to Mat\'ern-type dependence on irregular domains and meshes.

\subsection{Replicated fields via block operators}

Many data sets consist of repeated realizations of the same underlying dependence structure. Let $\mathbf{W}^{(r)}\in\mathbb R^{n}$ denote the latent field vector for replicate $r\in\{1,\ldots,R\}$, and stack
\[
\boldsymbol{\mathcal{W}} := \big((\mathbf{W}^{(1)})^\top,\ldots,(\mathbf{W}^{(R)})^\top\big)^\top \in \mathbb R^{Rn}.
\]
If each replicate follows the same operator equation
\[
\mathbf{K}(\boldsymbol{\theta})\,\mathbf{W}^{(r)} = \boldsymbol{\epsilon}^{(r)},
\qquad r=1,\ldots,R,
\]
then the stacked system can be written compactly as
$\mathbf{\mathcal{K}}(\boldsymbol{\theta})\,\boldsymbol{\mathcal{W}} = \boldsymbol{\mathcal{\epsilon}}$, with 
$\mathbf{\mathcal{K}}(\boldsymbol{\theta}) := \mathbf{I}_R \otimes \mathbf{K}(\boldsymbol{\theta})$,
where $\boldsymbol{\mathcal{\epsilon}}:=((\boldsymbol{\epsilon}^{(1)})^\top,\ldots,(\boldsymbol{\epsilon}^{(R)})^\top)^\top$ and $\otimes$ denotes the
Kronecker product. Since $\mathbf{K}(\boldsymbol{\theta})$ is sparse, $\mathbf{\mathcal{K}}(\boldsymbol{\theta})$ is a sparse matrix with a block-diagonal
structure. 
\change{Crucially, because the replicates share the same underlying spatial or temporal mesh, the vector of mesh weights $\mathbf{h} = (h_1, \dots, h_n)^\top$ is also replicated. The stacked weight vector can be elegantly expressed as $\mathbf{1}_R \otimes \mathbf{h}$ (where $\mathbf{1}_R$ is an $R$-dimensional vector of ones). This ensures that the local variance scaling for the stacked innovations $\boldsymbol{\mathcal{\epsilon}}$ remains mathematically consistent across all $R$ realizations.}




\subsection{Tensor product models}
\label{sec:tensor-product-models}

Tensor product constructions provide a general and powerful framework for building high-dimensional models from simpler, lower-dimensional components. Let $\mathbf{W}$ be a latent field discretized so that the stacked latent vector has dimension $n = \prod_{i=1}^d n_i$. For a composite model formed from $d$ components, let $\mathbf{K}_i(\boldsymbol{\theta}_i) \in \mathbb{R}^{n_i \times n_i}$ denote the sparse Markov operator matrices associated with each dimension $i$. The global latent operator matrix is constructed via the Kronecker product:
\begin{equation}
\mathbf{K}(\boldsymbol{\theta}) = \mathbf{K}_d(\boldsymbol{\theta}_d) \otimes \dots \otimes \mathbf{K}_1(\boldsymbol{\theta}_1),
\label{eq:K-kronecker}
\end{equation}
acting on the stacked coefficient vector $\mathbf{W} \in \mathbb{R}^n$.

Conditional on the mixing variables, the latent component is Gaussian and the operator equation takes the form
\[
\mathbf{K}(\boldsymbol{\theta})\,\mathbf{W} = \boldsymbol{\epsilon},
\qquad
\boldsymbol{\epsilon}\mid \mathbf{V} \sim \mathcal{N}\!\big(\mathbf{m}(\mathbf{V}),\mathbf{D}_\mathbf{V}\big),
\]
where $\mathbf{D}_\mathbf{V}$ is the diagonal matrix of conditional variances. Throughout, we employ the centered mean specification
\[
\mathbf{m}(\mathbf{V})=\mu(\mathbf{V}-\mathbf{h}), \qquad E[\mathbf{V}]=\mathbf{h},
\]
so that $E[\boldsymbol{\epsilon}]=\mathbf{0}$.
Consistent with the tensor product structure of the operator, the multi-dimensional mesh weights are also formed multiplicatively. By identifying the component mesh weights as column vectors $\mathbf{h}_i \in \mathbb{R}^{n_i}$, the combined weight column vector is given by the Kronecker product $\mathbf{h} = \mathbf{h}_d \otimes \dots \otimes \mathbf{h}_1$. This ensures that the local variance of the innovations $\boldsymbol{\epsilon}$ is properly scaled by the full multi-dimensional cell volumes.

The conditional precision matrix of $\mathbf{W}$ is defined as
\[
\mathbf{Q}(\mathbf{V},\boldsymbol{\theta}) = \mathbf{K}(\boldsymbol{\theta})^\top \mathbf{D}_\mathbf{V}^{-1} \mathbf{K}(\boldsymbol{\theta}).
\]
Consequently, this conditional precision matrix does not generally admit a Kronecker factorization. However, in spatial statistics, separability is typically defined in terms of the \emph{unconditional covariance} (i.e., second-order structure). By the law of total covariance,
\[
\mathrm{Cov}(\boldsymbol{\epsilon})
=
E\!\left[\mathrm{Cov}(\boldsymbol{\epsilon}\mid \mathbf{V})\right]
+
\mathrm{Cov}\!\left(E[\boldsymbol{\epsilon}\mid \mathbf{V}]\right)
=
E[\mathbf{D}_{\mathbf{V}}] + \mathrm{Cov}\!\big(\mathbf{m}(\mathbf{V})\big).
\]
Under the centered mean specification $\mathbf{m}(\mathbf{V})=\mu(\mathbf{V}-\mathbf{h})$, this becomes
\[
\mathrm{Cov}(\boldsymbol{\epsilon}) = E[\mathbf{D}_{\mathbf{V}}] + \mu^2\,\mathrm{Cov}(\mathbf{V}).
\]
Assuming the operator matrices $\mathbf{K}_i$ are square and invertible, the unconditional covariance of $\mathbf{W}$ is
\begin{equation}
\boldsymbol{\Sigma}
=
\mathrm{Cov}(\mathbf{W})
=
\mathbf{K}(\boldsymbol{\theta})^{-1}
\Big(E[\mathbf{D}_{\mathbf{V}}] + \mu^2\,\mathrm{Cov}(\mathbf{V})\Big)
\mathbf{K}(\boldsymbol{\theta})^{-T}.
\label{eq:Sigma-uncond-general}
\end{equation}

Consistent with the tensor product construction, the mesh weights factorize as
\[
\mathrm{diag}(\mathbf{h})
=
\mathrm{diag}(\mathbf{h}_d \otimes \cdots \otimes \mathbf{h}_1)
=
\bigotimes_{i=d}^1 \mathrm{diag}(\mathbf{h}_i).
\]
In the common variance-mixture discretization based on independent cell-wise mixing variables, $E[\mathbf{D}_{\mathbf{V}}]\propto \mathrm{diag}(\mathbf{h})$, and moreover $\mathrm{Cov}(\mathbf{V})$ is diagonal with entries proportional to the same weights, i.e., $\mathrm{Cov}(\mathbf{V})\propto \mathrm{diag}(\mathbf{h})$. Hence
\[
E[\mathbf{D}_{\mathbf{V}}] + \mu^2\,\mathrm{Cov}(\mathbf{V})
\ \propto\ 
\mathrm{diag}(\mathbf{h})
=
\bigotimes_{i=d}^1 \mathrm{diag}(\mathbf{h}_i),
\]
and the unconditional covariance is \emph{covariance-separable} (second-order separable). Using the mixed-product property of Kronecker products, we obtain the explicit factorization
\begin{equation}
\boldsymbol{\Sigma} \propto \bigotimes_{i=d}^1
\left(
\mathbf{K}_i(\boldsymbol{\theta}_i)^{-1}\,
\mathrm{diag}(\mathbf{h}_i)\,
\mathbf{K}_i(\boldsymbol{\theta}_i)^{-T}
\right).
\label{eq:Sigma-separable}
\end{equation}
Thus, while the conditional precision \(\mathbf{Q}(\mathbf{V},\boldsymbol{\theta})\) is generally non-separable, the tensor product model is separable in the standard \emph{second-order} sense even under variance-mixture non-Gaussianity.

In the specific case of homogeneous Gaussian innovations where $\mathbf{D}_\mathbf{V}=\mathbf{I}$ and $\mathbf{m}(\mathbf{V})\equiv\mathbf{0}$, the precision simplifies explicitly to the separable form
\[
\mathbf{Q}(\boldsymbol{\theta})
=
\mathbf{K}(\boldsymbol{\theta})^\top \mathbf{K}(\boldsymbol{\theta})
=
\bigotimes_{i=d}^1
\big(\mathbf{K}_i(\boldsymbol{\theta}_i)^\top \mathbf{K}_i(\boldsymbol{\theta}_i)\big).
\]

A particularly important example of this approach is the separable space--time model ($d=2$). Let $\mathbf{W}(s,t)$ be a latent field discretized so that the stacked vector $\mathbf{W}$ has dimension $n=n_sn_t$. Letting $\mathbf{K}_s(\boldsymbol{\theta}_s)\in\mathbb R^{n_s\times n_s}$ and $\mathbf{K}_t(\boldsymbol{\theta}_t)\in\mathbb R^{n_t\times n_t}$ denote the sparse spatial and temporal Markov operator matrices respectively, the separable space--time latent operator matrix follows directly from Equation~\ref{eq:K-kronecker} as:
\begin{equation}
\mathbf{K}_{\mathrm{st}}(\boldsymbol{\theta})
=
\mathbf{K}_t(\boldsymbol{\theta}_t)\ \otimes\ \mathbf{K}_s(\boldsymbol{\theta}_s).
\label{eq:Kst-kronecker}
\end{equation}

Following the general framework, the combined space--time mesh weight column vector is $\mathbf{h}_{\mathrm{st}} = \mathbf{h}_t \otimes \mathbf{h}_s$. While the mixing vector $\mathbf{V}$ does not factorize as $\mathbf{V}_t \otimes \mathbf{V}_s$, its first- and second-order moments inherit the tensor product weights under the independent-increments construction, so that
\[
E[\mathbf{D}_{\mathbf{V}}] + \mu^2\,\mathrm{Cov}(\mathbf{V})
\ \propto\ 
\mathrm{diag}(\mathbf{h}_{\mathrm{st}})
=
\mathrm{diag}(\mathbf{h}_t \otimes \mathbf{h}_s)
=
\mathrm{diag}(\mathbf{h}_t)\otimes \mathrm{diag}(\mathbf{h}_s),
\]
and the unconditional covariance separates into temporal and spatial components as in \eqref{eq:Sigma-separable}. For homogeneous Gaussian innovations ($\mathbf{D}_\mathbf{V}=\mathbf{I}$), the conditional precision matrix explicitly reduces to the separable space--time structure:
\begin{equation}
\mathbf{Q}_{\mathrm{st}}(\boldsymbol{\theta})
=
\big(\mathbf{K}_t(\boldsymbol{\theta}_t)^\top \mathbf{K}_t(\boldsymbol{\theta}_t)\big)\ \otimes\ \big(\mathbf{K}_s(\boldsymbol{\theta}_s)^\top \mathbf{K}_s(\boldsymbol{\theta}_s)\big).
\label{eq:Qst-from-Kst}
\end{equation}

\subsection{Non-separable space--time models via SPDEs}
\label{subsec:manymodels-nonseparable}

Separable space--time constructions (e.g., $\mathbf{K}_t\otimes \mathbf{K}_s$) are often effective baselines, but they
cannot represent genuine space--time interaction, such as transport or directionality, where temporal
evolution depends on spatial gradients and the resulting covariance is typically \emph{non-separable}
(and may be asymmetric in time and/or space).
A principled operator-first route to such interaction is to start from an \emph{evolution SPDE} that
encodes physically interpretable mechanisms and then discretize it into a sparse GMRF on the
space--time grid; see, e.g., the advection--diffusion SPDE construction in \cite{advection_diffusion}.

Let $X(\mathbf{s},t)$ be a spatio-temporal field on $\mathbb{R}^d\times\mathbb{R}$.
A flexible class of non-separable models is obtained from the first-order-in-time advection--diffusion
SPDE
\begin{equation}
\label{eq:advection-diffusion-spde}
\Bigg[
\frac{\partial}{\partial t}
+\frac{1}{c}\big(\kappa^2-\nabla\cdot \mathbf{H}\nabla\big)
+\frac{1}{c}\,\boldsymbol{\gamma}^\top \nabla
\Bigg] X(\mathbf{s},t)
=
\frac{\tau}{\sqrt c}\,Z(\mathbf{s},t),
\end{equation}
where $c>0$ is a temporal scaling parameter, $\kappa>0$ controls the spatial range, $\mathbf{H}$ is a
diffusion/anisotropy matrix, $\boldsymbol{\gamma}$ is an advection (transport) vector, and $Z(\mathbf{s},t)$ is a
spatio-temporal forcing noise.

A key point for LLnGM-style computation is that \eqref{eq:advection-diffusion-spde} admits a sparse GMRF
approximation after standard operator discretizations: (i) an implicit Euler finite difference scheme for the time derivative, and (ii) a finite element discretization of the spatial operators at each time step.
This yields a sparse block-lower-bidiagonal operator matrix $\mathbf{K}(\boldsymbol{\theta})$ acting on the stacked latent state $\mathbf{W} = (\mathbf{W}_1^\top, \ldots, \mathbf{W}_T^\top)^\top$:

\begin{equation}
\mathbf{K}_{\mathrm{adv-diff}}(\boldsymbol{\theta}) = 
\begin{pmatrix}
\mathbf{K}_{1,1} & \mathbf{0} & \mathbf{0} & \cdots & \mathbf{0} \\
\mathbf{K}_{2,1} & \mathbf{K}_{2,2} & \mathbf{0} & \cdots & \mathbf{0} \\
\mathbf{0} & \mathbf{K}_{3,2} & \mathbf{K}_{3,3} & \cdots & \mathbf{0} \\
\vdots & \vdots & \ddots & \ddots & \vdots \\
\mathbf{0} & \mathbf{0} & \cdots & \mathbf{K}_{T,T-1} & \mathbf{K}_{T,T}
\end{pmatrix}.
\label{eq:spacetime-K-matrix}
\end{equation}

For $t=2,\dots,T$, the blocks are defined by the implicit Euler update:
\begin{align*}
\mathbf{K}_{t,t} &= \sqrt{c}\,\mathbf{C} + \frac{1}{\sqrt{c}}\,\mathbf{L}_t, \\
\mathbf{K}_{t,t-1} &= -\sqrt{c}\,\mathbf{C}.
\end{align*}
Here, $\mathbf{C}$ is the spatial mass matrix ($C_{ij} = \langle \psi_i, \psi_j \rangle$), and $\mathbf{L}_t$ represents the discretized spatial operator (including diffusion, advection, and stabilization) at time $t$:
\[
\mathbf{L}_t = \kappa^2 \mathbf{C} + \mathbf{G} + \mathbf{B}(\boldsymbol{\gamma}_t) + \mathbf{S}(\boldsymbol{\gamma}_t),
\]
where $\mathbf{G}$ is the stiffness matrix, $\mathbf{B}(\boldsymbol{\gamma}_t)$ is the advection matrix, and $\mathbf{S}(\boldsymbol{\gamma}_t)$ is the streamline diffusion stabilization matrix.
This structure allows for computationally efficient inference using sparse matrix algorithms despite the non-separable nature of the model.
\change{The construction of $\mathbf{h}$ is same as in the separable space-time model case.}

\begin{table}[htb]
\centering
\caption{Examples of Latent Dependence Operators $\K(\theta)$ in the LLnGM Framework.
All Constructions Yield Sparse Operators after Discretization and Admit Scalable Conditional Gaussian Inference.
}
\label{tab:operator-matrices}
\begin{adjustbox}{max width=0.90\textwidth}
\begin{tabular}{>{\raggedright\arraybackslash}m{3.3cm}
                >{\centering\arraybackslash}m{5.2cm}
                >{\raggedright\arraybackslash}m{6.2cm}}
\toprule
\textbf{Construction type} & \textbf{Operator $\K(\theta)$} & \textbf{Interpretation} \\
\midrule

\textbf{i.i.d. latent field}
& $\mathbf{I}_n$
& Independent latent effects; no structured dependence. \\

\midrule

\textbf{Temporal Markov models}
& $\K_t(\theta_t)$ (banded)
& Linear difference or recursion operators (AR($p$), RW($k$), OU) inducing Markov dependence in time. \\

\midrule

\textbf{Spatial SPDE (Mat\'ern)}
& $\kappa^2 \C + \G$
& Finite element discretization of $(\kappa^2 - \Delta)^{\alpha/2} w = \mathcal W$ \citep{lindgren2011}. \\

\midrule

\textbf{Replicated fields}
& $\mathbf{I}_R \otimes \K_0(\theta)$
& $R$ independent replicates sharing the same latent operator $\K_0$. \\

\midrule

\textbf{Multivariate / coupled fields}
& $(\mathbf{D}(\boldsymbol{\psi})\otimes I)\,
   \mathrm{blockdiag}(\K_1,\ldots,\K_q)$
& Operator-level mixing across components; $\K_j$ control marginal dependence and $D(\psi)$ cross-correlation \citep{multivariate_matern}. \\

\midrule

\textbf{Separable space--time}
& $\K_t(\theta_t)\otimes \K_s(\theta_s)$
& Kronecker construction combining temporal and spatial operators; yields separable space--time dependence. \\

\midrule

\textbf{Non-separable space--time}
& $\K_{\mathrm{adv-diff}}(\theta)$
& Discretized evolution SPDEs (e.g., advection--diffusion) coupling space and time through differential operators \citep{advection_diffusion}. \\



\bottomrule
\end{tabular}
\end{adjustbox}
\end{table}


\section{Details of the inference algorithm}\label{app:inference-details}

Here we give the specific forms of the conditional distributions used in the Gibbs sampling algorithm in Algorithm~\ref{algo:gibbs}.

\subsection{The conditional distributions used in the Gibbs sampling}

\begin{proposition}[Conditional distribution of \(\mathbf{W}|\mathbf{V}, \mathbf{Y}\)]
\label{prp:conditional-w}
Under the LLnGM model \eqref{eq:framework}, the conditional distribution of \(\mathbf{W}|\mathbf{V}, \mathbf{Y}\) is given by
  $\mathbf{W}|\mathbf{V}, \mathbf{Y} \sim \mathcal{N}(\boldsymbol{\mu}_W, \boldsymbol{\Sigma}_W)$,
where
$\boldsymbol{\Sigma}_W^{-1} = \mathbf{K}^{\top} \mathbf{D}_{\mathbf{V}^W}^{-1} \mathbf{K} + \mathbf{A}^{\top} \mathbf{D}_{\mathbf{V}^Y}^{-1} \mathbf{A}$, 
$$
\boldsymbol{\mu}_W = \boldsymbol{\Sigma}_W \left( \mathbf{K}^{\top} \mathbf{D}_{\mathbf{V}^W}^{-1} \boldsymbol{\mu}^W \odot (\mathbf{V}^W - \mathbf{h}) + \mathbf{A}^{\top} \mathbf{D}_{\mathbf{V}^Y}^{-1} (\mathbf{Y} - \mathbf{X}\boldsymbol{\beta} - \boldsymbol{\mu}^Y \odot (\mathbf{V}^Y - \mathbf{1})) \right),
$$
and \(\mathbf{D}_{\mathbf{V}^W} = \text{diag}(\boldsymbol{\sigma}^W \odot \boldsymbol{\sigma}^W \odot \mathbf{V}^W)\), \(\mathbf{D}_{\mathbf{V}^Y} = \text{diag}(\boldsymbol{\sigma}^Y \odot \boldsymbol{\sigma}^Y \odot \mathbf{V}^Y)\).
\end{proposition}

\begin{proof}
Condition on $(\mathbf V^Y,\mathbf V^W,\mathbf Y)$ and treat $(\boldsymbol\theta,\boldsymbol\beta)$ as fixed.
From \eqref{eq:data}--\eqref{eq:noise}, the observation equation implies
\begin{equation}
\mathbf Y \mid \mathbf W,\mathbf V^Y
\sim \mathcal N\!\Big(\mathbf A\mathbf W+\mathbf X\boldsymbol\beta+\mathbf m^{Y},\ \mathbf D_{\mathbf V^Y}\Big),
\label{eq:lik_y_given_w}
\end{equation}
where
$\mathbf m^{Y}=\boldsymbol\mu^{Y}\odot(\mathbf V^{Y}-\mathbf 1)$, and 
$\mathbf D_{\mathbf V^Y}=\text{diag}\!\big((\boldsymbol\sigma^{Y}\odot\boldsymbol\sigma^{Y})\odot \mathbf V^{Y}\big)$.
Similarly, the process equation \eqref{eq:process} and the noise specification \eqref{eq:noise} give
\[
\mathbf K\mathbf W = \boldsymbol\epsilon^{(2)},\qquad
\boldsymbol\epsilon^{(2)}\mid \mathbf V^W \sim \mathcal N(\mathbf m^{W},\mathbf D_{\mathbf V^W}),
\]
with
$\mathbf m^{W}=\boldsymbol\mu^{W}\odot(\mathbf V^{W}-\mathbf h)$, and 
$\mathbf D_{\mathbf V^W}=\text{diag}\!\big((\boldsymbol\sigma^{W}\odot\boldsymbol\sigma^{W})\odot \mathbf V^{W}\big)$.
Therefore, by linear transformation of a Gaussian random vector,
\begin{equation}
\mathbf K\mathbf W \mid \mathbf V^W \sim \mathcal N(\mathbf m^{W},\mathbf D_{\mathbf V^W}),
\label{eq:lik_kw}
\end{equation}
equivalently,
\begin{equation}
\mathbf K\mathbf W-\mathbf m^{W} \mid \mathbf V^W \sim \mathcal N(\mathbf 0,\mathbf D_{\mathbf V^W}).
\label{eq:lik_kw_centered}
\end{equation}

Since $\mathbf Y$ depends on $\mathbf W$ only through \eqref{eq:lik_y_given_w}, and the process model
contributes the Gaussian term \eqref{eq:lik_kw_centered}, the conditional density of $\mathbf W$ given
$(\mathbf V,\mathbf Y)$ is proportional to the product of these two Gaussian kernels:
\begin{align}
p(\mathbf W\mid \mathbf V,\mathbf Y)
&\propto
\exp\!\left\{
-\tfrac12(\mathbf Y-\mathbf A\mathbf W-\mathbf X\boldsymbol\beta-\mathbf m^{Y})^\top
\mathbf D_{\mathbf V^Y}^{-1}
(\mathbf Y-\mathbf A\mathbf W-\mathbf X\boldsymbol\beta-\mathbf m^{Y})
\right\}
\nonumber\\
&\quad\times
\exp\!\left\{
-\tfrac12(\mathbf K\mathbf W-\mathbf m^{W})^\top
\mathbf D_{\mathbf V^W}^{-1}
(\mathbf K\mathbf W-\mathbf m^{W})
\right\}.
\label{eq:post_kernel}
\end{align}
Taking $-2\log$ of \eqref{eq:post_kernel} and expanding the quadratic forms, then collecting the
terms that depend on $\mathbf W$, yields
\begin{align*}
-2\log p(\mathbf W\mid \mathbf V,\mathbf Y)
&=
\mathbf W^\top\!\Big(\mathbf A^\top \mathbf D_{\mathbf V^Y}^{-1}\mathbf A
+ \mathbf K^\top \mathbf D_{\mathbf V^W}^{-1}\mathbf K\Big)\mathbf W \\
&\quad
-2\,\mathbf W^\top\!\Big(
\mathbf A^\top \mathbf D_{\mathbf V^Y}^{-1}(\mathbf Y-\mathbf X\boldsymbol\beta-\mathbf m^{Y})
+\mathbf K^\top \mathbf D_{\mathbf V^W}^{-1}\mathbf m^{W}
\Big)
+\text{const},
\end{align*}
where $\text{const}$ does not depend on $\mathbf W$. This is the canonical quadratic form of a multivariate
normal density. Hence,
$\mathbf W\mid \mathbf V,\mathbf Y \sim \mathcal N(\boldsymbol\mu_W,\boldsymbol\Sigma_W)$,
with precision matrix
$\boldsymbol\Sigma_W^{-1}
=\mathbf K^\top \mathbf D_{\mathbf V^W}^{-1}\mathbf K
+\mathbf A^\top \mathbf D_{\mathbf V^Y}^{-1}\mathbf A$,
and mean vector
$\boldsymbol\mu_W
=\boldsymbol\Sigma_W\left(
\mathbf K^\top \mathbf D_{\mathbf V^W}^{-1}\mathbf m^{W}
+\mathbf A^\top \mathbf D_{\mathbf V^Y}^{-1}(\mathbf Y-\mathbf X\boldsymbol\beta-\mathbf m^{Y})
\right)$.
The stated expression is obtained by substituting $\mathbf m^{W}=\boldsymbol\mu^{W}\odot(\mathbf V^{W}-\mathbf h)$ and
$\mathbf m^{Y}=\boldsymbol\mu^{Y}\odot(\mathbf V^{Y}-\mathbf 1)$.
\end{proof}


\begin{proposition}[Conditional distributions of mixing variables]
\label{prp:conditional-v}
Under the LLnGM model, the conditional distributions of the mixing variables are:
\begin{align*}
V^W_i | \mathbf{W}, \mathbf{Y} &\sim \text{GIG}\left(p_i^W - \frac{1}{2}, \quad a_i^W + \frac{(\mu^W_i)^2}{(\sigma^W_i)^2}, \quad b_i^W + \frac{((\mathbf{K}\mathbf{W})_i + \mu^W_i h_i^W)^2}{(\sigma^W_i)^2}\right), \\
V^Y_i | \mathbf{W}, \mathbf{Y} &\sim \text{GIG}\left(p_i^Y - \frac{1}{2}, \quad a_i^Y + \frac{(\mu^Y_i)^2}{(\sigma^Y_i)^2}, \quad b_i^Y + \frac{(Y_i - (\mathbf{X}\boldsymbol{\beta} + \mathbf{A}\mathbf{W})_i + \mu^Y_i)^2}{(\sigma^Y_i)^2}\right).
\end{align*}
\end{proposition}

\begin{proof}
We give the proof for $V_i^W$; the case $V_i^Y$ follows by the same argument.

Fix $i$ and write $r_i^W := (\mathbf K\mathbf W)_i$. Conditional on $V_i^W=v$, the process noise model
\eqref{eq:noise} implies
\[
r_i^W=\epsilon^{(2)}_i,\qquad 
\epsilon^{(2)}_i\mid v \sim \mathcal N\!\big(\mu_i^W(v-h_i),\ \sigma_i^W v\big),
\]
and the prior is $v\sim\text{GIG}(p_i^W,a_i^W,b_i^W)$. Hence, up to a multiplicative constant independent of $v$,
\[
p(v\mid \mathbf W,\mathbf Y)
\ \propto\
(\sigma_i^W v)^{-1/2}
\exp\!\left\{-\frac{\big(r_i^W-\mu_i^W(v-h_i)\big)^2}{2\sigma_i^W v}\right\}
\cdot
v^{p_i^W-1}\exp\!\left\{-\frac12\Big(a_i^W v+b_i^W/v\Big)\right\}.
\]
Therefore
\[
p(v\mid \mathbf W,\mathbf Y)
\ \propto\
v^{p_i^W-\frac32}
\exp\!\left\{
-\frac12\Big(a_i^W v+b_i^W/v\Big)
-\frac{\big(r_i^W-\mu_i^W(v-h_i)\big)^2}{2\sigma_i^W v}
\right\}.
\]
Expanding the quadratic term,
$r_i^W-\mu_i^W(v-h_i)=(r_i^W+\mu_i^W h_i)-\mu_i^W v$,
so that
\[
\frac{\big(r_i^W-\mu_i^W(v-h_i)\big)^2}{v}
=
\frac{(r_i^W+\mu_i^W h_i)^2}{v}
-2\mu_i^W(r_i^W+\mu_i^W h_i)
+(\mu_i^W)^2 v.
\]
The middle term does not depend on $v$ and is absorbed into the normalizing constant. Thus
\[
p(v\mid \mathbf W,\mathbf Y)
\ \propto\
v^{p_i^W-\frac32}
\exp\!\left\{
-\frac12\left(
\Big(a_i^W+\frac{(\mu_i^W)^2}{\sigma_i^W}\Big)v
+
\Big(b_i^W+\frac{(r_i^W+\mu_i^W h_i)^2}{\sigma_i^W}\Big)\frac1v
\right)
\right\},
\]
which is the kernel of a $\text{GIG}\!\left(p_i^W-\frac12,\ a_i^W+\frac{(\mu_i^W)^2}{\sigma_i^W},\
b_i^W+\frac{(r_i^W+\mu_i^W h_i)^2}{\sigma_i^W}\right)$ distribution. Hence
\[
V_i^W\mid \mathbf W,\mathbf Y \sim 
\text{GIG}\!\left(p_i^W-\frac12,\ a_i^W+\frac{(\mu_i^W)^2}{\sigma_i^W},\
b_i^W+\frac{\big((\mathbf K\mathbf W)_i+\mu_i^W h_i\big)^2}{\sigma_i^W}\right).
\]

For $V_i^Y$, let $r_i^Y:=Y_i-(\mathbf X\boldsymbol\beta+\mathbf A\mathbf W)_i$. Using
\[
r_i^Y=\epsilon^{(1)}_i,\qquad \epsilon^{(1)}_i\mid v \sim \mathcal N\!\big(\mu_i^Y(v-1),\ \sigma_i^Y v\big),
\qquad v\sim\text{GIG}(p_i^Y,a_i^Y,b_i^Y),
\]
the same calculation yields
\begin{align*}
V_i^Y\mid \mathbf W,\mathbf Y &\sim 
\text{GIG}\!\left(p_i^Y-\frac12,\ a_i^Y+\frac{(\mu_i^Y)^2}{\sigma_i^Y},\
b_i^Y+\frac{\big(r_i^Y+\mu_i^Y\big)^2}{\sigma_i^Y}\right) \\
&=
\text{GIG}\!\left(p_i^Y-\frac12,\ a_i^Y+\frac{\big(Y_i-(\mathbf X\boldsymbol\beta+\mathbf A\mathbf W)_i+\mu_i^Y\big)^2}{\sigma_i^Y},\
b_i^Y+\frac{(\mu_i^Y)^2}{\sigma_i^Y}\right).
\end{align*}
This proves the claimed conditional forms.
\end{proof}

\subsection{The formulation of the gradient estimator}\label{sec:gradient}

In this section, we will derive the gradient estimator of the joint log-likelihood function for the LLnGM model in Section~\ref{chp:framework}, utilizing Fisher's identity \citep[Proposition D.4]{douc2013nonlinear}.

By the discussion in \cite{wallin2015}, for \(V_i\) to have a known parametric distribution for any \(h_i\), we need that the variance process to belong to a class of distributions that is closed under convolution, and there are only two special cases of the GH distribution that has this property \citep{podgorski2016}, which is NIG and GAL distributions.
For both the NIG and the GAL distributions, the formulation is over-parameterized. Therefore, we will adapt the parameterization in \cite{asar2020}, i.e., consider \(V_i \sim \text{Inv-Gaussian}(\nu, \nu h_i^2)\) for NIG driven noise and \(V_i \sim \text{Gamma}(h_i \nu, \nu)\) for GAL driven noise, \(\nu > 0\) is the parameter of the variance component distribution.

Let $\mathbf{Y} = (y_1, y_2, \dots, y_m)^\top$, $\mathbf{W} = (w_1, w_2, \dots, w_n)^\top$, and $\bs{\theta} = (\theta_1, \theta_2, \dots, \theta_p)^\top$ denote the vector of kernel parameters. The joint log-likelihood function is given by
\[
    \log \pi(\mathbf{Y}, \mathbf{W}, \mathbf{V}) = \log \pi(\mathbf{Y} | \mathbf{W}) + \log \pi(\mathbf{W} | \mathbf{V}) + \log \pi(\mathbf{V}),
\]
where the components are:
\begin{align*}
    \log \pi(\mathbf{Y}| \mathbf{W}) &= -\frac{m}{2} \log(2\pi) - m \log \sigma_{\epsilon} - \frac{1}{2\sigma_{\epsilon}^2} \|\mathbf{Y} - \mathbf{X} \bs{\beta} - \mathbf{A} \mathbf{W} \|^2, \\
    \log \pi(\mathbf{W} | \mathbf{V}) &= -\frac{n}{2} \log(2\pi) - n \log \sigma + \log \det(\mathbf{K}) - \frac{1}{2} \log \det(\mathbf{D}_{\mathbf{V}}) \\
    & \quad - \frac{1}{2\sigma^2} (\mathbf{K}\mathbf{W} - \mu (\mathbf{V}-\mathbf{h}))^\top \mathbf{D}_{\mathbf{V}}^{-1} (\mathbf{K}\mathbf{W} - \mu (\mathbf{V}-\mathbf{h})),
\end{align*}
with $\mathbf{D}_{\mathbf{V}} = \text{diag}(\mathbf{V})$. The term $\log \pi(\mathbf{V})$ depends on the mixing distribution. For the case  
 $V_i \sim \text{Inv-Gaussian}(\nu, \nu h_i^2)$:
        \[
          \log \pi(\mathbf{V}) = \sum_{i=1}^n \left( \frac{1}{2} \log \left( \frac{\nu h_i^2}{2 \pi V_i^3} \right) - \frac{\nu V_i^2 + \nu h_i^2}{2 V_i} + \nu h_i \right).
        \]
    For $V_i \sim \text{Gamma}(h_i \nu, \nu)$:
        \[
          \log \pi(\mathbf{V}) = \sum_{i=1}^n \left( \log \left( \frac{\exp(-V_i \nu) (V_i \nu)^{h_i \nu} }{V_i \Gamma(h_i \nu)} \right) \right).
        \]

We now derive the partial derivatives of the joint log-likelihood,  $\mathcal{L} = \log \pi(\mathbf{Y}, \mathbf{W}, \mathbf{V})$ with respect to each parameter.

\paragraph*{Regression coefficients $\bs{\beta}$}
The gradient with respect to the regression coefficients is
\[
    \nabla_{\bs{\beta}} \mathcal{L} = \frac{1}{\sigma_{\epsilon}^2} \mathbf{X}^\top (\mathbf{Y} - \mathbf{X} \bs{\beta} - \mathbf{A} \mathbf{W}).
\]

\paragraph*{Noise variance $\sigma_{\epsilon}$}
Differentiating with respect to the observation noise variance yields
\[
    \partial_{\sigma_{\epsilon}} \mathcal{L} = -\frac{m}{\sigma_{\epsilon}} + \frac{1}{\sigma_{\epsilon}^3} \|\mathbf{Y} - \mathbf{X} \bs{\beta} - \mathbf{A} \mathbf{W} \|^2.
\]

\paragraph*{Process scale $\sigma$}
Expanding the quadratic term, we obtain:
\begin{align*}
    \partial_\sigma \mathcal{L} 
    & = -\frac{n}{\sigma} + \frac{1}{\sigma^3} \sum_{i=1}^n \frac{((\mathbf{K}\mathbf{W})_i - \mu (V_i - h_i))^2}{V_i} \\
    & = -\frac{n}{\sigma} + \frac{1}{\sigma^3} \sum_{i=1}^n \left( - 2\mu (\mathbf{K}\mathbf{W})_i - 2h_i \mu^2 + \mu^2 V_i + \frac{(\mathbf{K}\mathbf{W})_i^2}{V_i}  - \frac{2 h_i \mu (\mathbf{K} \mathbf{W})_i}{V_i} - \frac{h_i^2 \mu^2}{V_i} \right).
\end{align*}

\paragraph*{Skewness parameter $\mu$}
The gradient with respect to the skewness parameter is
\begin{align*}
    \partial_\mu \mathcal{L}
    & = \frac{1}{\sigma^2} \sum_{i=1}^n \frac{((\mathbf{K}\mathbf{W})_i - \mu (V_i - h_i))(V_i - h_i)}{V_i} \\ 
    & = \frac{1}{\sigma^2} \sum_{i=1}^n \left( (\mathbf{K}\mathbf{W})_i + 2h_i \mu - \mu V_i - \frac{h_i^2 \mu}{V_i} - \frac{h_i (\mathbf{K} \mathbf{W})_i}{V_i} \right).
\end{align*}

\paragraph*{Kernel parameters $\theta_i$}
For each kernel parameter $\theta_i$ in $\mathbf{K}(\bs{\theta})$, the gradient is given by
\[
    \partial_{\theta_i} \mathcal{L} = \text{tr}(\mathbf{K}^{-1} \partial_{\theta_i} \mathbf{K}) - \frac{1}{\sigma^2} \mathbf{W}^\top (\partial_{\theta_i} \mathbf{K})^\top \mathbf{D}_{\mathbf{V}}^{-1} (\mathbf{K}\mathbf{W} - \mu (\mathbf{V}-\mathbf{h})).
\]

\paragraph*{Mixing parameter $\nu$}
The gradient with respect to the mixing parameter depends on the choice of distribution:
\[ 
\partial_\nu \mathcal{L} = 
\begin{cases}
\sum_{i=1}^n \left( - \frac{h_i^2}{2 V_i} - \frac{V_i}{2} + \frac{1}{2\nu} + h_i \right) & \text{for NIG}\\
\sum_{i=1}^n \left( h_i - V_i + h_i \log(V_i) + h_i \log(\nu) - h_i \psi(h_i \nu) \right) & \text{for GAL}
\end{cases}
\]
where $\psi(\cdot)$ is the digamma function.

The MC gradient estimator follows by substituting unobserved \( \W \) and \(\V\) from the Gibbs samples. The RB gradient estimator follows by substituting \(\V\) from the Gibbs samples and \(\W\) from the conditional distribution of \(\W|\V, \Y\) in Proposition \ref{prp:conditional-w}.

\section{Comparative Study: \texttt{ngme2} vs. TMBStan and Stan}
\label{app:comparison}

We conducted a comparative study to evaluate the performance of \texttt{ngme2} against established Bayesian frameworks: Template Model Builder (TMB) integrated with Stan's NUTS sampler (\texttt{tmbstan}) and standard Stan MCMC. The comparison focused on an AR(1) process with Normal-Inverse Gaussian (NIG) noise to assess computational efficiency, convergence robustness, and parameter recovery.

We simulated $N=500$ observations from an AR(1) model with NIG innovations:
\begin{align}
W_t &= \rho W_{t-1} + \epsilon_t, \quad \epsilon_t \sim \text{NIG}(\mu, \sigma, \nu) \\
Y_t &= W_t + \eta_t, \quad \eta_t \sim \mathcal{N}(0, \sigma_\epsilon^2)
\end{align}
The true parameters were set to $\mu = 3$, $\sigma = 2$, $\rho = 0.8$, $\sigma_\epsilon = 1$, and $\nu = 0.4$. 
The NIG distribution is represented as a normal variance-mean mixture: $\epsilon_t \mid V_t \sim \mathcal{N}(\mu(V_t - 1), \sigma^2 V_t)$, where $V_t \sim \text{IG}(\nu, \nu)$.

The implementation details are as follows:

\begin{itemize}
    \item \textbf{\texttt{ngme2}}:  Rao-Blackwellized SGD optimization, which converged within 400 iterations, followed by posterior exploration using SGLD to draw 2,000 samples.
    \change{The traceplot of the MAP estimation using \texttt{ngme2} is shown in Figure~\ref{fig:traceplot-ngme2}.}
    \item \textbf{Stan MCMC}: A full hierarchical implementation in Stan, treating latent variables $\{W_t, V_t\}$ as parameters. We ran 4 chains with 2,000 iterations each (1,000 warmup).
    \change{The traceplot of the Stan MCMC is shown in Figure~\ref{fig:traceplot-stan-mcmc}.}
    \item \textbf{TMBStan}: A hybrid approach where TMB marginalized the latent variables $\{W_t, \log V_t\}$ via Laplace approximation. The resulting marginal log-posterior was then sampled using Stan's NUTS algorithm. 
    \change{The traceplot of the TMBStan is shown in Figure~\ref{fig:traceplot-tmbstan}.}
\end{itemize}

\begin{figure}[t]
\centering
\includegraphics[width=0.8\textwidth]{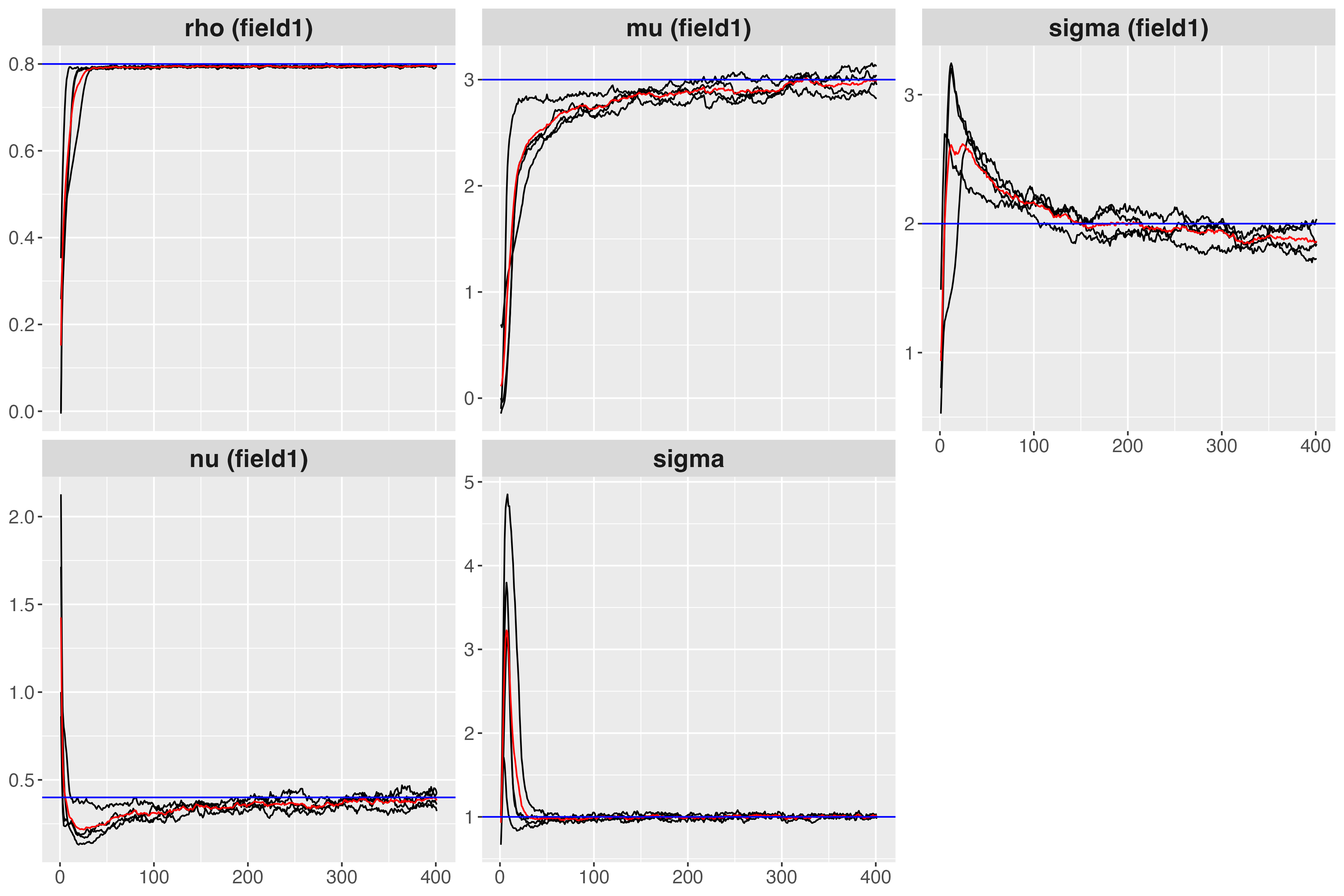}
\caption{Traceplots of the MAP Estimation for the Parameters in the NIG-AR(1) Model using \texttt{ngme2}.
The red lines are the average of 4 parallel runs. 
The blue lines indicate the true values.
}
\label{fig:traceplot-ngme2}
\end{figure}

\begin{figure}[t]
\centering
\includegraphics[width=0.8\textwidth]{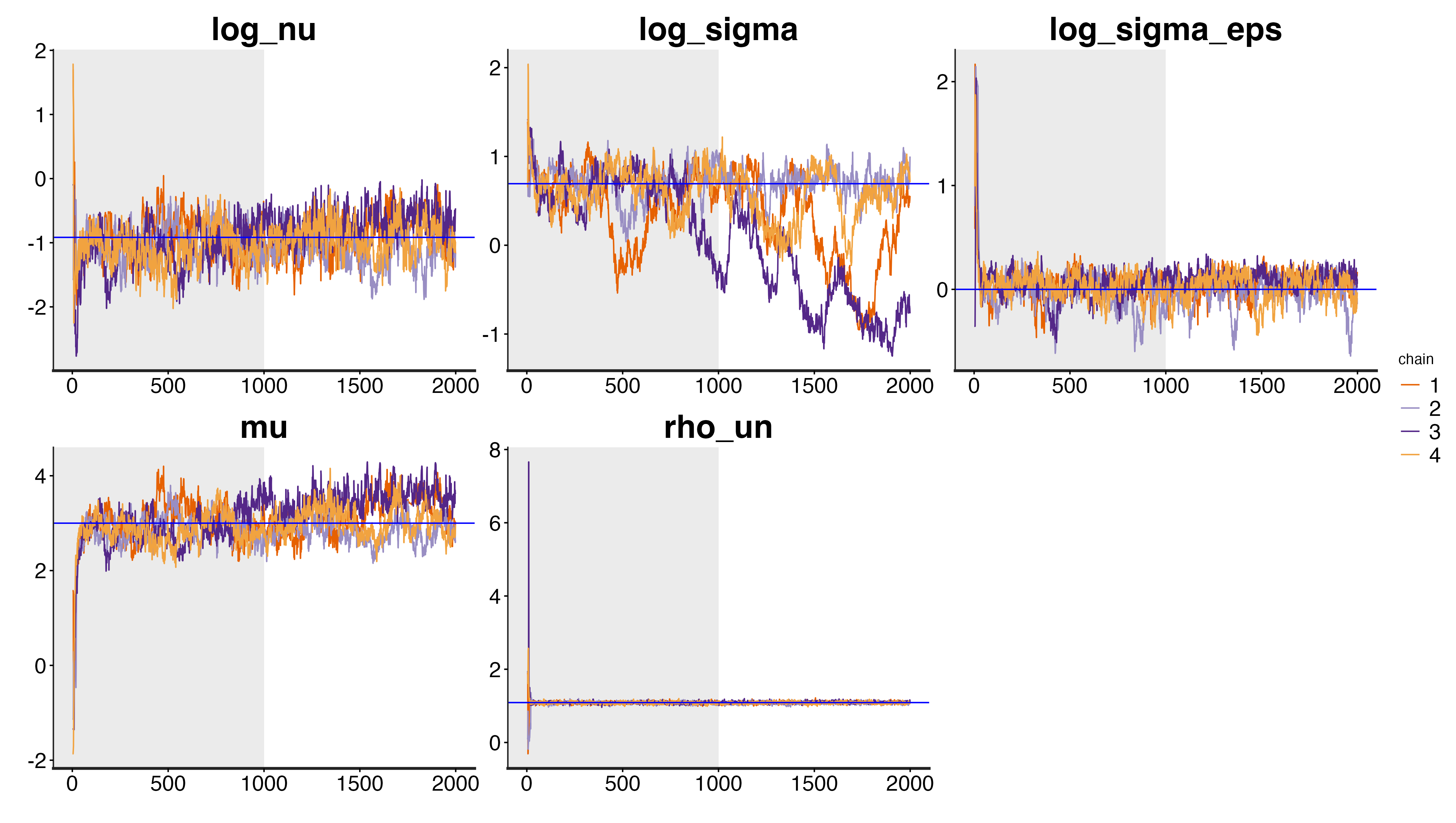}
\caption{Traceplots of the Parameters in the NIG-AR(1) Model using \texttt{stan}. The blue lines indicate the true values.}
\label{fig:traceplot-stan-mcmc}
\end{figure}

\begin{figure}[t]
\centering
\includegraphics[width=0.8\textwidth]{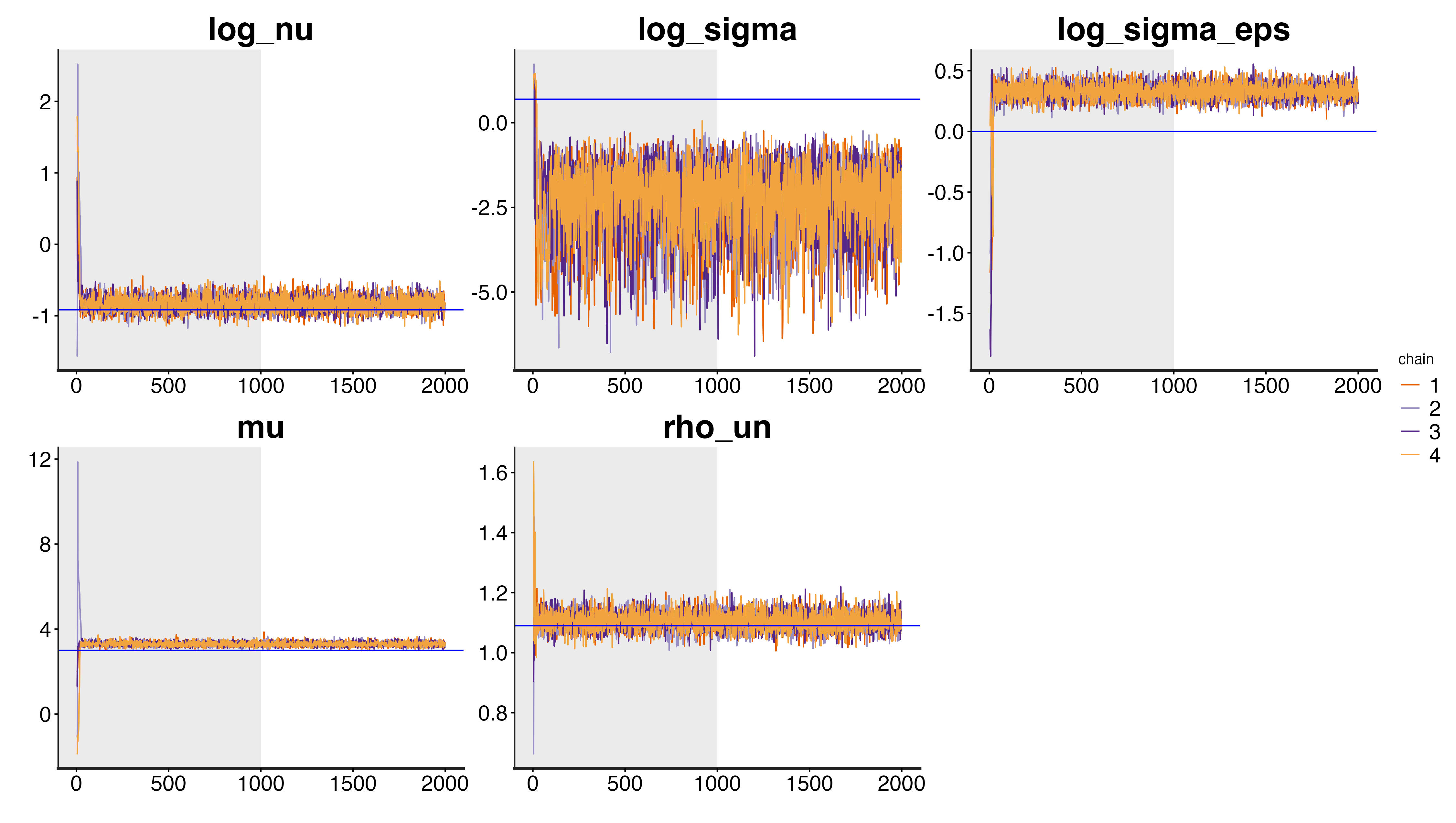}
\caption{Traceplots of the Parameters in the NIG-AR(1) Model using \texttt{tmbstan}. The blue lines indicate the true values.}
\label{fig:traceplot-tmbstan}
\end{figure}

To enforce the stationarity constraint $\rho\in(-1,1)$ during the optimization, we introduce an unconstrained parameter
$\phi\in\mathbb{R}$ defined by
\(
\phi = \log\!\left(\frac{1+\rho}{1-\rho}\right),
\)
which is the inverse of the transform $\rho=\tanh(\phi/2)$.
Table~\ref{tab:tmbstan-comparison} summarizes the computational time, convergence metrics, and the Kullback-Leibler divergence for each method.
Table~\ref{tab:parameter-estimates} reports the parameter estimates.
Tables~\ref{tab:stan-mcmc-summary}--\ref{tab:tmbstan-summary} report standard MCMC summary statistics:
\texttt{mean} and \texttt{sd} are the posterior mean and standard deviation; \texttt{mcse} is the Monte Carlo
standard error of the reported mean estimate; the columns \texttt{2.5\%}, \texttt{25\%}, \texttt{50\%}, \texttt{75\%},
and \texttt{97.5\%} are posterior quantiles; $n_{\text{eff}}$ is the effective sample size, and $\hat{R}$ is the Gelman--Rubin potential scale reduction factor, with values close to 1
indicating good mixing across chains.

\begin{table}[t]
\centering
\caption{Comparison of Estimation Methods for AR(1)-NIG Model}
\label{tab:tmbstan-comparison}
\begin{tabular}{lrrr}
\toprule
\textbf{Method} & \textbf{Time (seconds)} & \textbf{Convergence} & \textbf{Noise KLD} \\
\midrule
ngme2 & 20.0 & Converged (SGD-based rule) & 0.011 \\
Stan MCMC & 37.7 & Poor ($\hat{R} = 2.11$) & 0.075 \\
TMBStan & 360.0 & Converged ($\hat{R} = 1.00$) & 2.026 \\
\bottomrule
\end{tabular}
\end{table}

\paragraph{Computational Efficiency} 
\texttt{ngme2} exhibited the highest efficiency, completing the estimation in 20.0 seconds. This represents a roughly 2-fold speedup over full Stan MCMC (37.7 seconds) and nearly a 17-fold speedup over TMBStan (360.0 seconds). While TMBStan benefits from marginalization, the need to evaluate the dense Hessian for the Laplace approximation at every HMC leapfrog step introduces a massive computational bottleneck. The iterative stochastic optimization in \texttt{ngme2} proves significantly faster and more scalable for this class of latent non-Gaussian models.

\paragraph{Convergence and Mixing} 
As shown in Table~\ref{tab:stan-mcmc-summary}, standard Stan MCMC struggled with the high-dimensional latent space, resulting in divergent transitions, poor mixing (Max $\hat{R} = 2.11$), and an extremely low effective sample size ($n_{\text{eff}} = 4$ for $\log \sigma$). TMBStan successfully resolved these sampling pathologies by marginalizing the latent field via Laplace approximation, yielding perfect chain mixing ($\hat{R} = 1.00$) and high effective sample sizes (Table~\ref{tab:tmbstan-summary}). However, as discussed below, this came at the cost of severe estimation bias. In contrast, \texttt{ngme2} achieved stable convergence without requiring approximations that compromise accuracy.

\paragraph{Estimation Accuracy} 
Since the correlation parameter $\rho$ is well estimated across all methods, we evaluated the accuracy of the noise recovery using the Kullback-Leibler (KL) divergence between the true and estimated distributions of the driving noise. \texttt{ngme2} yielded the closest fit to the ground truth (KLD = 0.011). As seen in Table~\ref{tab:parameter-estimates}, Stan MCMC failed to accurately recover the scale parameter $\sigma$ due to poor posterior exploration. Meanwhile, TMBStan's Laplace approximation forced a Gaussian geometry onto the heavy-tailed NIG latent field. This resulted in a severely collapsed latent variance ($\hat{\sigma} \approx 0.10$) and an artificially inflated measurement error ($\hat{\sigma}_{\epsilon} \approx 1.39$), completely failing to recover the true driving noise distribution (KLD = 2.026).

\begin{table}[t]
\centering
\caption{Parameter Estimates Comparison}
\label{tab:parameter-estimates}
\begin{tabular}{lrrrrr}
\toprule
\textbf{Method} & $\boldsymbol{\mu}$ & $\boldsymbol{\sigma}$ & $\boldsymbol{\rho}$ & $\boldsymbol{\sigma_\epsilon}$ & $\boldsymbol{\nu}$ \\
\midrule
True & 3.000 & 2.000 & 0.800 & 1.000 & 0.400 \\
ngme2 & 3.035 & 1.718 & 0.795 & 1.001 & 0.362 \\
Stan MCMC & 3.140 & 1.264 & 0.796 & 1.041 & 0.409 \\
TMBStan & 3.299 & 0.096 & 0.802 & 1.393 & 0.433 \\
\bottomrule
\end{tabular}
\end{table}

\subsection{Discussion}

The results demonstrate the significant challenges traditional methods encounter when modeling heavy-tailed and asymmetric NIG-driven latent fields. Standard Stan MCMC fails to mix efficiently because of the challenging funnel geometry and strong dependence between the latent variables and variance hyperparameters. 
While the TMB-Stan hybrid successfully mitigates the mixing issues by marginalizing the latent state, it introduces a critical flaw: the underlying Laplace approximation relies on a Gaussian assumption around the mode. This assumption fundamentally fails to capture the complexity of the non-Gaussian latent field, leading to collapsed latent variances and highly biased estimates, all while incurring a massive computational cost. By leveraging stochastic optimization with Rao-Blackwellized Monte Carlo gradients, \texttt{ngme2} avoids both the sampling pathologies of HMC and the approximation biases of Laplace methods, confirming that its optimization-based framework is superior in scalability, robustness, and precision for complex LLnGMs.

\begin{table}[t]
\centering
\caption{Posterior Summary for Stan MCMC}
\label{tab:stan-mcmc-summary}
\begin{adjustbox}{max width=\textwidth}
\begin{tabular}{lrrrrrrrrrr}
\toprule
 & mean & mcse & sd & 2.5\% & 25\% & 50\% & 75\% & 97.5\% & $n_{\text{eff}}$ & $\hat{R}$ \\
\midrule
$\mu$            & 3.140 & 0.163 & 0.392 & 2.471 & 2.828 & 3.121 & 3.431 & 3.921 & 6 & 1.49 \\
$\log(\sigma)$    & 0.235 & 0.301 & 0.589 & -0.982 & -0.179 & 0.442 & 0.715 & 0.951 & 4 & 2.11 \\
$\phi$       & 1.089 & 0.003 & 0.032 & 1.026 & 1.067 & 1.089 & 1.110 & 1.152 & 127 & 1.04 \\
$\log(\sigma_{\epsilon})$ & 0.040 & 0.025 & 0.130 & -0.285 & -0.024 & 0.060 & 0.130 & 0.236 & 28 & 1.12 \\
$\log(\nu)$       & -0.894 & 0.089 & 0.311 & -1.510 & -1.109 & -0.888 & -0.672 & -0.313 & 12 & 1.25 \\
\bottomrule
\end{tabular}
\end{adjustbox}
\end{table}

\begin{table}[t]
\centering
\caption{Posterior Summary for TMBStan (with Laplace Approximation)}
\label{tab:tmbstan-summary}
\begin{adjustbox}{max width=\textwidth}
\begin{tabular}{lrrrrrrrrrr}
\toprule
 & mean & mcse & sd & 2.5\% & 25\% & 50\% & 75\% & 97.5\% & $n_{\text{eff}}$ & $\hat{R}$ \\
\midrule
$\mu$            & 3.299 & 0.002 & 0.099 & 3.116 & 3.233 & 3.294 & 3.361 & 3.503 & 3594 & 1.00 \\
$\log(\sigma)$    & -2.341 & 0.021 & 1.057 & -4.870 & -2.989 & -2.175 & -1.559 & -0.726 & 2571 & 1.00 \\
$\phi$       & 1.105 & 0.000 & 0.030 & 1.047 & 1.085 & 1.105 & 1.124 & 1.164 & 4148 & 1.00 \\
$\log(\sigma_{\epsilon})$ & 0.332 & 0.001 & 0.062 & 0.209 & 0.291 & 0.331 & 0.372 & 0.455 & 4102 & 1.00 \\
$\log(\nu)$       & -0.837 & 0.002 & 0.101 & -1.024 & -0.909 & -0.840 & -0.770 & -0.631 & 3737 & 1.00 \\
\bottomrule
\end{tabular}
\end{adjustbox}
\end{table}

\section*{Acknowledgments}

This publication is based upon work supported by King Abdullah University of Science and Technology (KAUST) under Award No. ORFS-CRG11-2022-5015.

\bibliographystyle{abbrvnat}
\bibliography{References}

\end{document}